\documentclass{gOMS2e}

\usepackage{epstopdf}
\usepackage{algorithm}
\usepackage[noend]{algpseudocode}
\usepackage{caption}
\usepackage{bm}
\usepackage{mathtools}
\usepackage{subfig}
\usepackage{graphicx}	

\theoremstyle{plain}
\newtheorem{theorem}{Theorem}[section]

\newtheorem{lemma}[theorem]{Lemma}
\newtheorem{claim}[theorem]{Claim}

\theoremstyle{definition}

\theoremstyle{remark}
\newtheorem{remark}{Remark}

\makeatletter
\def\BState{\State\hskip-\ALG@thistlm}
\makeatother

\mathtoolsset{showonlyrefs=true}

\newcommand{\postcov}{\Gamma_{\!\text{\scriptsize post}} }
\newcommand{\prcov} {\Gamma_{\!\text{\scriptsize prior}} }
\newcommand{\ncov}{\Gamma_{\!\text{\scriptsize noise}}}
\newcommand{\parm}{u_0}
\newcommand{\parvec}{\hat{u}_0}
\newcommand{\indom}{\Omega_{\!\text{\scriptsize in}}}
\newcommand{\obs}{u}
\newcommand{\obsvec}{\hat{u}}
\newcommand{\pto}{\mathcal{F}}
\newcommand{\outdom}{\Omega_{\!\text{\scriptsize out}}}

\newcommand{\R}{\mathbb{R}}
\newcommand{\diag}{\operatorname{diag}}
\newcommand{\mO}{\mathcal{O}}
\newcommand{\prcovI}{\sigma^2_{\!\text{\scriptsize prior}}\mathrm{I}_m}
\newcommand{\ncovI}{\sigma^2_{\!\text{\scriptsize noise}}\mathrm{I}_n}
\newcommand{\prsigma}{\sigma^{-2}_{\!\text{\scriptsize prior}}}
\newcommand{\nsigma}{\sigma^{-2}_{\!\text{\scriptsize noise}}}
\newcommand{\rsigma}{\sigma^2_{\!\text{\scriptsize noise}}/\sigma^2_{\!\text{\scriptsize prior}}}
\newcommand{\Neach}{N_{\!\text{\scriptsize each}}}
\newcommand{\neach}{n_{\!\text{\scriptsize each}}}

\begin{document}



\title{Solving Optimal Experimental Design with Sequential Quadratic Programming and Chebyshev Interpolation}

\author{
\name{Jing Yu\textsuperscript{a}$^{\ast}$\thanks{$^\ast$Corresponding author. Email: jingyu621@uchicago.edu}
and Mihai Anitescu\textsuperscript{a}\textsuperscript{b}}
\affil{\textsuperscript{a}The University of Chicago, 5747 S. Ellis Avenue, Chicago, IL 60637 USA; \break
\textsuperscript{b}Argonne National Laboratory, 9700 South Cass Ave, Lemont IL 60439 USA.}
\received{July 2019}
}

\maketitle

\begin{abstract}
We propose an optimization algorithm to compute the optimal sensor locations in experimental design in the formulation of Bayesian inverse problems, where the parameter-to-observable mapping is described through an integral equation and its discretization results in a continuously indexed matrix whose size depends on the mesh size $n$. By approximating the gradient and Hessian of the objective design criterion from Chebyshev interpolation, we solve a sequence of quadratic programs and achieve the complexity $\mO(n\log^2(n))$. An error analysis guarantees the integrality gap shrinks to zero as $n\to\infty$, and we apply the algorithm on a two-dimensional advection-diffusion equation, to determine the LIDAR's optimal sensing directions for data collection.
\end{abstract}


\begin{classcode}
65C60, 90-08, 65C20
\end{classcode}

\section{Introduction}\label{sec:intro}

An important branch of experimental design attempts to compute the optimal sensor locations given a set of available measurement points (\citep[\S7.5]{boyd} and \citep[\S9, \S12]{puke}), with the aim to give the most accurate estimation of parameters or maximize the information about a system. Naturally it arises in many infrastructure networks (oil, water, gas, and
electricity) in which large amounts of sensor data need to be processed in real time in order to reconstruct the state of the system or to identify leaks, faults, or attacks (see \citep{sandia2}, \citep{krausewater}, \citep{sandia3} and \citep{yu2018scalable}).

We consider the same setting as in \citep{sur_oed}, where the statistical setup consists of a Bayesian framework for linear inverse problems for which the direct relationship is described by discretized integral equation. Specifically, the parameter $\parm$ is a function that maps the input domain $\indom$ to $\R$, the observable $\obs$ is a function from the output domain $\outdom$ to $\R$, and the parameter-to-observable mapping $\pto:\parm\to\obs$ is given by an integral equation:
\begin{equation}\label{intro_eq:pto}
\obs(x) = \int_{\indom} f(x, y)\parm(y)\,\mathrm{d}y, \quad\mbox{for }x\in\outdom.
\end{equation}
Both $\indom$ and $\outdom$ are rectangular domains. The two sets of points $\{x_1, x_2,\cdots,x_n\}\subset\outdom$ and $\{y_1, y_2,\cdots,y_m\}\subset\indom$ are the discretized mesh points, and the goal is to estimate the parameter vector $\parvec=(\parm(y_1),\cdots,\parm(y_m))$ as a proxy for the unknown function $\parm$. The observable vector $\obsvec=(\obs(x_1),\cdots,\obs(x_n))$ represents data points that can be potentially observed by sensors. The number of sensors is limited and we need to select the optimal locations from the candidate set $\{x_1,\cdots,x_n\}$.

Following \citep{sur_oed}, we use a Bayesian framework with Gaussian priors and likelihood, assign the weight $w_i$ to each candidate location $x_i$, and formulate the optimal sensor placement as a convex integer program:
\begin{equation}\label{intro_opt:doe}
\begin{array}{rl}
\mbox{min}_w\quad &\phi(\postcov(w))\\
\mbox{s.t.}\quad  &w_i\in\{0, 1\},\ \sum_{i=1}^n w_i=n_0
\end{array}
\end{equation}
where $\phi$ can be either A-, D- or E-optimal design criterion, corresponding to the trace, log determinant and largest eigenvalue respectively, of the posterior covariance matrix $\postcov$ for the parameter $\parvec$. $\postcov$ is computed from Bayes' rule and is given by
\begin{equation}
\postcov = (F^TW^{1/2}\ncov^{-1}W^{1/2}F + \prcov^{-1} )^{-1}.
\end{equation}
where $W=\diag(w_1,w_2,\cdots,w_n)$. $F$ is the discretized parameter-to-observable mapping, and $F(i, j) = f(x_i, y_j)\Delta y$, where $\Delta y$ is the size of unit rectangle on $\indom$, i.e. $\Delta y = \mu(\indom)/m$. It has been shown in \citep{sur_oed} that, if $\ncov$ and $\prcov$ are multiples of identity matrices ($\prcov=\prcovI, \ncov=\ncovI$), and the ratios $m/n$, $n_0/n$ are constants, then by solving the relaxed optimization
\begin{equation}\label{intro_opt:doerelax}
\begin{array}{rl}
\mbox{min}_w\quad &\phi(\postcov(w))\\
\mbox{s.t.}\quad  &0\le w_i\le1,\ \sum_{i=1}^n w_i=n_0
\end{array}
\end{equation}
and applying the \textit{sum-up rounding} (SUR) strategy (see Appendix \ref{app:sur}), the integrality gap between an upper bound of \eqref{intro_opt:doe} obtained from the SUR solution, and a lower bound of \eqref{intro_opt:doe} obtained from the relaxed solution to \eqref{intro_opt:doerelax}, converges to zero as the mesh sizes $n\to\infty$.

In this paper, we provide an optimization algorithm to solve the relaxation \eqref{intro_opt:doerelax}, based on Chebyshev interpolation and sequential quadratic programming (SQP). While the relaxed problem is not NP-hard, computing its gradient requires $\mO(n^3)$ operations, and finding the Hessian is even more expensive. Given that $f(x,y)$ typically comes from a mathematical model described by a system of partial differential equations (PDEs), and the parameters to be estimated are initial or boundary conditions, the discretization of an increasingly refined mesh can easily explode the problem size to thousands and even millions. A $\mO(n^3)$ algorithm is intractable in practice, and we need a scalable algorithm to solve \eqref{intro_opt:doerelax} both fast and accurately.

Several efficient algorithms have been proposed to alleviate the computation burden for specific design criteria (see \citep{alexanderian2014optimal, PetraMartinStadlerEtAl14, anderian2018efficient}), all of which exploit the low-rank structure of the parameter-to-observable mapping $\pto$ in some way. In \citep{alexanderian2014optimal}, randomized methods, such as randomized singular value decomposition (rSVD) and randomized trace estimator, are employed to evaluate the A-optimal design objective, i.e. the trace of $\postcov$, and its gradient. Its approximation error depends on the threshold chosen in rSVD and the sample size in randomized estimations. In \citep{anderian2018efficient}, similar approaches (truncated spectral decomposition, randomized estimators for determinants) are investigated for the D-optimal design, and the related error bounds are derived explicitly. In addition to the  numerical results, a general study on the optimal low-rank update from the prior covariance matrix to the posterior covariance matrix, is available over a broad class of loss functions (see \citep{spantini2015optimal}).

We make use of the integral operator assumption and the fact $F$ is continuously indexed i.e., $F(i,j)$ is evaluated from a smooth function $f(x,y)$, and propose an interpolation-based method to approximate both the gradient and Hessian for A- and D-optimal designs. Classical interpolation theory is well established for function approximations with function evaluations only at a subset of points, and it is known that polynomial interpolation at Chebyshev points is optimal in the minimax error for continuously differentiable functions (see \citep[\S8.5]{ greenbaum2012numerical}). We apply Chebyshev interpolation to extract the gradient and Hessian information, and implement sequential quadratic programming to compute the relaxed solution, where each quadratic program is solved by an interior-point algorithm. The advantage of SQP is threefold: in contrast to previous methods, we incorporate Hessian to accelerate the convergence rate in the optimization algorithm; we are able to prove the zero convergence of the approximation error in the objective as the problem size approaches infinity; the overall complexity is $\mO(n\log^2(n))$. 

In the numerical experiment, we apply this algorithm to a two-dimensional LIDAR problem, aiming to select the optimal sensing directions for the LIDAR to send beams and collect data along those beams, in order to infer the initial condition of the advection-diffusion equation. While the designs rely on the constants in the system, we can safely conclude more sensing directions should be chosen towards the velocity field in the equation. We demonstrate the algorithm efficiency by providing the computation time and the shrinkage of the integrality gap.

We focus on A-optimal design in the paper, and save results on D-optimal design to the Appendix, which are derived very similarly. In summary, we discuss the gradient and Hessian approximations for the trace objective in \S\ref{sec:interp}, and provide details of the SQP algorithm in \S\ref{sec:sqp}. An error analysis on the zero convergence of integrality gap is given in \S\ref{sec:error}, and in \S\ref{sec:lidar} we introduce and illustrates the designs for a LIDAR problem. Finally in \S\ref{sec:disc} we discuss ways to improve the current algorithm.

\section{Gradient and Hessian Approximations with Chebyshev Interpolation}\label{sec:interp}
Since $m/n$ is a constant, without loss of generality we assume $m=n$, and rewrite the posterior covariance matrix:
\begin{equation}
\postcov = \Big(\nsigma F^TWF + \prsigma\mathrm{I}_n \Big)^{-1} = \sigma^2_{\!\text{\scriptsize noise}}\Big(F^TWF + \alpha \mathrm{I}_n \Big)^{-1}
\end{equation}
where $\alpha = \rsigma$. We ignore the constants $\prsigma$, $\nsigma$ at the moment and add them back in the numerical section. In the A-optimal design, $\phi(\postcov) = tr(\postcov)$, and we provide the analytical form of the derivatives below, and then approximate them by exploiting the continuously indexed structure.

\subsection{Gradient and Hessian for A-optimal Design}
Here are the component-wise gradient and Hessian of the objective function $tr(\postcov)$:
\begin{itemize}
\item \textbf{Gradient}
Denote $f_i$ as the $i$-th column of $F^T$ and we have
\begin{equation*}
F^TWF = \sum_{i=1}^n w_i f_if_i^T \quad\Rightarrow\quad \frac{\partial F^TWF}{\partial w_i} = f_if_i^T.
\end{equation*}
Therefore the $i^{th}$ component in the gradient is:
\begin{equation}\label{interp_eq:grad}
\frac{\partial tr(\postcov)}{\partial w_i} = -tr\Big((F^TWF+\mathrm{I}_n)^{-1}f_if_i^T(F^TWF+\mathrm{I}_n)^{-1} \Big)= -\|(F^TWF+\mathrm{I}_n)^{-1}f_i \|^2.
\end{equation}
\item \textbf{Hessian}
Following the previous steps, the $(i,j)^{th}$ entry of Hessian matrix is:
\begin{equation}\label{interp_eq:Hessian}
H_{ij} = \frac{\partial^2 tr(\postcov)}{\partial w_i\partial w_j} = 2\Big( f_i^T(F^TWF + \mathrm{I}_n)^{-1} f_j \Big)\Big( f_i^T(F^TWF + \mathrm{I}_n)^{-2} f_j \Big).
\end{equation}
\end{itemize}
Note that $f_i$ is discretized from a smooth function $f(x_i, \cdot)$, so both $f_i$ and $F$ are continuously indexed. In addition, $H$ is the Schur product of two positive semi-definite matrices $F(F^TWF+\mathrm{I}_n)^{-1}F^T$ and $F(F^TWF+\mathrm{I}_n)^{-2}F^T$, so $H$ is also positive semi-definite (see \citep[\S1.5]{zhang2006schur}). Next we explain an approximation of the gradient and Hessian with Chebyshev interpolation.

\subsection{Chebyshev Interpolation in One Dimension}
The $N$ Chebyshev interpolation points on $[-1, 1]$ are 
\begin{equation}
\tilde{x}_i = \cos\Big(\frac{\pi (i - 1)}{N - 1}\Big), \quad i = 1,2,..,N.
\end{equation}
Given a smooth function $h: [-1, 1]\to\mathbb{R}$, and evaluations at the interpolation points $\{(\tilde{x}_1, h(\tilde{x}_1)), (\tilde{x}_2, h(\tilde{x}_2)), .., (\tilde{x}_{N}, h(\tilde{x}_{N}))\}$, then for any $x\in[-1, 1]$, we approximate $h(x)$ with Lagrange basis polynomials:
\begin{equation}
\tilde{h}(x)=\sum_{i=1}^{N} \big(\prod_{j\neq i}\frac{x - \tilde{x}_j}{\tilde{x}_i-\tilde{x}_j}\big)h(\tilde{x}_i).
\end{equation}
Chebyshev interpolation points achieve the minimal error $\|h-\tilde{h}\|_{\infty}$ among polynomial approximations, and we will discuss its accuracy in \S\ref{sec:error}. We choose the number of interpolation points as $N = \mathcal{O}(log(n))$ for both computational and accuracy purposes, and create the coefficient vector associated with $x$ as
\begin{equation}
c(x) = \Big(\prod_{j\neq 1}\frac{x - \tilde{x}_1}{\tilde{x}_i-\tilde{x}_1}, \prod_{j\neq 2}\frac{x - \tilde{x}_2}{\tilde{x}_i-\tilde{x}_2}, ..., \prod_{j\neq N}\frac{x - \tilde{x}_N}{\tilde{x}_i-\tilde{x}_N} \Big)^T\in\mathbb{R}^N.
\end{equation}

To construct a low-rank approximation of $F$, we create a matrix $\tilde{F}\in\mathbb{R}^{N\times N}$ with
\begin{equation}
\tilde{F}(i,j) = f(x_i, y_j)\Delta y
\end{equation}
where $\tilde{x}_i$ and $\tilde{y}_j$ are interpolation points in $\outdom$ and $\indom$ respectively. Define the coefficient matrix $C_x = \big(c_x(x_1), c_x(x_2),...,c_x(x_n)\big)\in\mathbb{R}^{N\times n}$ for each mesh point $x_i\in\indom$, and define $C_y$ in similar way. Then $F$ is approximated by  

\begin{equation}\label{interp_eq:F_s}
F_s:=C_x^T\tilde{F}C_y.
\end{equation}

To approximate the gradient and Hessian in \eqref{interp_eq:grad} and \eqref{interp_eq:Hessian}, we construct $M\in\mathbb{R}^{n\times N}$ with its $i$-th column $m_i$ given by
\begin{equation}\label{interp_eq:col_F}
\Big(F_s^TWF_s + \mathrm{I}_n\big)^{-1}\tilde{f}_i,
\end{equation}
where $\tilde{f}_i$ is the $i^{th}$ column of $C_y^T\tilde{F}^T$, as an approximation of the column in $F^T$ evaluated at $\tilde{x}_i$. 
We then define $M_1, M_2\in\mathbb{R}^{N\times N}$ where the $(i,j)^{th}$ entry is $\langle\tilde{f}_i, m_j\rangle$ and $\langle m_i, m_j\rangle$ respectively, that is, for $i,j=1,2,...,N$,
\begin{equation}
M_1(i,j) = \tilde{f}_i^T\Big(F_s^TWF_s + \mathrm{I}_n\big)^{-1}\tilde{f}_j, \quad M_2(i,j) =  \tilde{f}_i^T\Big(F_s^TWF_s+ \mathrm{I}_n\big)^{-2}\tilde{f}_j.
\end{equation}

\begin{itemize}
\item \textbf{Approximate gradient in \eqref{interp_eq:grad}.} Let $g\in\mathbb{R}^n$ be the true gradient, i.e.
\begin{equation}
g = (\frac{\partial \phi}{\partial w_1}, \frac{\partial \phi}{\partial w_2}, ..., \frac{\partial \phi}{\partial w_n})^T.
\end{equation}
and we approximate each component by $g_i\approx -c_x(x_i)^TM_2c_x(x_i)$.

\item \textbf{Approximate Hessian in \eqref{interp_eq:Hessian}.} We construct another matrix $\tilde{H}$ where $\tilde{H}(i,j) = 2M_1(i,j)M_2(i,j)$, then $H(i,j)$ is approximated by $c_x(x_i)^T \tilde{H}c_x(x_j)$. Equivalently,
\begin{equation}
H\approx H_s = C_x^T \tilde{H} C_x.
\end{equation}
Note $\tilde{H}$ is the Schur product of two positive semi-definite matrices $M_1$ and $M_2$, so $H_s$ is positive semi-definite (see \citep[\S1.5]{zhang2006schur}).
\end{itemize}

The above interpolation-based approximation can be generalized for any interval domain $[a, b]$ by defining a one-to-one mapping between $[a, b]$ and $[-1, 1]$.

\subsection{Chebyshev Interpolation in Two Dimensions} 
After we understand the interpolation approximation in one dimension, it is not difficult to extend it to multiple dimensions by tensor product, although the notation would be slightly more complicated. 

Consider the domain $\Omega = [-1,1]\times[-1,1]$, and let $\neach$ and $\Neach$ be the number of mesh points and interpolation points respectively on each side. We have $n=\neach^2$ mesh points and $N=\Neach^2$ interpolation points in total and they are related by
\begin{equation}
\Neach^2 = N   = \mathcal{O}(log(n)) = \mathcal{O}\big(log(\neach^2)\big).
\end{equation}

Suppose $\{(x_i, x_j)\}_{i,j = 1}^{\neach}$ are mesh points, and $\{(\tilde{x}_i, \tilde{x}_j)\}_{i,j = 1}^{\Neach}$ are interpolation points, and we construct $\tilde{F}\in\mathbb{R}^{N\times N}$ in a similar way as in one dimension. The $n$ mesh points are ordered as follows: for an index $k\in\{1,2,..,n\}$, we write 
$$k = (i-1)\cdot \neach + (j-1),$$ 
and it corresponds to the mesh point $(x_i, x_j)$ in $\Omega$. In other words, we arrange these mesh points ``column by column'', and the index $k$ is associated with  $(x_i, x_j)$. We apply the same ordering to interpolation points. Next we find the coefficient vector $c(x_i,x_j)\in\mathbb{R}^N$, i.e. how a general function $h(x_i,x_j)$ depends on the values at interpolation points. Based on results from one dimension, let $c(x_i), c(x_j)\in\mathbb{R}^{\Neach}$ be the one-dimensional coefficient vector for $x_i$ and $x_j$, and $k\in\{1,2,..,N\}$ with
\begin{equation}
k = (k_1-1)\cdot \Neach + (k_2-1),
\end{equation}
then the $k^{th}$ coefficient is given by
\begin{equation}
c(x_i, x_j)_k = c(x_i)_{k_1}c(x_j)_{k_2}.
\end{equation}
The $k^{th}$ component in $c(x_i, x_j)\in\mathbb{R}^N$ is the product of $k_1^{th}$ component in $c(x_i)$ and $k_2^{th}$ component in $c(x_j)$. We calculate the coefficient vector for each mesh point, and create matrices $C$, $M$, $M_1$ and $M_2$ in a similar fashion (details omitted). Gradient and Hessian are approximated in the same way as in one dimension.

\section{Sequential Quadratic Program}\label{sec:sqp}

Given the (approximated) gradient and Hessian in the previous section, we solve a sequence of quadratic program until convergence, where at each step, the objective is a quadratic Taylor polynomial evaluated at the current iterate. We mainly adopt the algorithm from \citep[\S 18.1]{Nocedal_book}, and before that, we instead look at an equivalent version to the relaxed optimization \eqref{intro_opt:doerelax}:

\begin{equation}\label{algo_eq:doe2}
\begin{array}{rl}
\mbox{min}\quad &\phi(\Gamma_{post}(w))\\
\mbox{s.t.}\quad  &0\le w_i\le 1,\ \sum_{i=1}^n w_i\le n_0.
\end{array}
\end{equation}
\begin{claim}
This program and the original program has the same minimal point.
\end{claim}
\begin{proof}
If $w\preceq w'$ ($w_i\le w'_i$ for each $i$), then $\Big( F^TWF + \mathrm{I}_n \Big)^{-1}\succeq\Big( F^TW'F + \mathrm{I}_n \Big)^{-1}$.
\end{proof}

\subsection{A Framework for SQP}

We give the details on the SQP algorithm: suppose at the $k$-th iteration, $(w^k, \lambda^k)$ are the primal and dual variable, we solve the following quadratic program
\begin{equation}\label{algo_eq:quadprog}
\begin{array}{rl}
\mbox{min}\quad &\phi_k + \nabla\phi_k^T\cdot p^k + \frac{1}{2}(p^k)^T\cdot\nabla^2_{ww}\mathcal{L}_k\cdot p^k\\
\mbox{s.t.}\quad  &-w^{k}_i\le p^k_i\le 1-w^{k}_i,\ , i= 1,2,...,n,\\
& \sum_{i=1}^n p^k_i\le n_0 - \sum_{i=1}^nw^k_i
\end{array}
\end{equation}
where $\nabla\phi_k$ is the gradient of the objective and $\nabla^2_{ww}\mathcal{L}_k$ is the Hessian of the Lagrangian, both evaluated at the current iterate $w^k$. As there are only linear constraints, we have $\nabla^2_{ww}\mathcal{L}_k = \nabla^2_{ww}\phi_k = H_k$. The problem \eqref{algo_eq:quadprog} can be simplified as:
\begin{equation}\label{eq:quadprog}
\begin{array}{rl}
\mbox{min}\quad & g^T\cdot p^k + \frac{1}{2}(p^k)^T\cdot H\cdot p^k\\
\mbox{s.t.} \quad & A\cdot p^k \ge b.
\end{array}
\end{equation}
where $g =  \nabla\phi_k$, $H = \nabla^2_{ww}\phi_k$, $A=\begin{pmatrix}
\mathrm{I}_n \\ -\mathrm{I}_n \\ -\bm{1}^T
\end{pmatrix}\in\mathbb{R}^{(2n+1)\times n}$, $b = \begin{pmatrix}
-w^k\\ w^k - \bm{1}\\ \bm{1}^Tw^k - n_0
\end{pmatrix}\in\mathbb{R}^{2n+1}$. 

\vspace*{0.2cm}
Both $g$ and $H$ are from Chebyshev approximations. The new iterate $w^{k+1}$ is updated by $w^k + \alpha_k p^k$ where $p^k$ is the solution to the quadratic program \eqref{eq:quadprog}, and $\alpha_k$ is the step length determined by backtracking line search (see \citep[Algorithm 3.1]{Nocedal_book}). We discuss details on solving the program \eqref{eq:quadprog} and getting its Lagrangian multipliers in the next subsection, but provide the SQP framework now in Algorithm \ref{sqp_algo:sqp}. 

\vspace*{0.4cm}
\begin{algorithm}
\caption{SQP with line search for Solving \eqref{algo_eq:doe2} ($c = 0.5, \xi = 10^{-3}$, $\epsilon$ is user defined)}\label{sqp_algo:sqp}
\begin{algorithmic}[1]
\BState \emph{choose an initial state $(w^0, \lambda^0)$; set $k\gets 0$}
\BState \textbf{repeat} until $\phi(\postcov(w^k)) - \phi(\postcov(w^{k+1})) <\epsilon$\label{algo_algo:stop_kkt}
\State $\quad$ evaluate $\nabla\phi_k, \nabla^2_{ww}\phi_k$ from Chebyshev interpolation;
\State $\quad$ solve the quadratic program \eqref{eq:quadprog} to obtain $(p^k, \lambda^{k+1})$;
\State $\quad\ \alpha_k = 1$
\State $\quad$ \textbf{while} $\phi(\Gamma_{post}(w^k + \alpha p^k)) > \phi(\Gamma_{post}(x^k)) + \xi\alpha_k\nabla\phi_k^Tp^k$
\State $\quad\quad$ $\alpha_k = c * \alpha_k$ 
\State $\quad$ \textbf{end (while)}
\State $\quad$ set $w^{k+1}\gets w^k+\alpha_k\cdot p^k$, $\lambda^{k+1}\gets \lambda^k+\alpha_k(\lambda^{k+1}-\lambda^k)$;
\BState \textbf{end (repeat)};
\end{algorithmic}
\end{algorithm}

In the backtracking line search step, we need to evaluate $\phi(\Gamma_{post}(w))$ which involves the trace of an inverse matrix of size $n\times n$, and we propose a SVD-based method with complexity $\mathcal{O}(n\log^2(n))$ for the evaluation. Recall that $F_s=C_x^T\tilde{F}C_y$ in \eqref{interp_eq:F_s} where $C_x, C_y\in\mathbb{R}^{N\times n}$, $\tilde{F}\in\mathbb{R}^{N\times N}$ and $N=\mathcal{O}(\log(n))$, then
\begin{equation}
F_s^TWF_s = C_y^T\tilde{F}^TC_xWC_x^T\tilde{F}C_y.
\end{equation}
To compute $\phi(\Gamma_{post}(w))$, we only need to find the eigenvalues $\{\lambda_i\}_{i=1}^n$ of $F_s^TWF_s$ since
\begin{equation}
\phi(\Gamma_{post}(w)) = tr\big((F_s^TWF_s+\mathrm{I}_n)^{-1}\big)=\sum_{i=1}^n\frac{1}{1 + \lambda_i}.
\end{equation}
We apply SVD decompositions to both $C_y^T\tilde{F}\in\mathbb{R}^{n\times N}$ and $C_xW^{1/2}\in\mathbb{R}^{N\times n}$, and get
\begin{equation}
C_y^T\tilde{F} = U_1\Sigma_1 V_1^T,\ C_xW^{1/2} = U_2\Sigma_2 V_2^T\quad\Rightarrow\quad F_s^TWF_s = U_1\Big(\Sigma_1 V_1^TU_2\Sigma_2^2U_2^TV_1\Sigma_1\Big)U_1^T.
\end{equation}
Because $\Sigma_1 V_1^TU_2\Sigma_2^2U_2^TV_1\Sigma_1\in\mathbb{R}^{N\times N}$ is of small size, another SVD decomposition (or eigenvalue decomposition) of this matrix directly gives us the eigenvalue decomposition of $F_s^TWF_s$, and thus the value of $\phi(\Gamma_{post}(w))$. Moreover, once we know $F_s^TWF_s = Q\Lambda Q^T$, where $Q\in\R^{n\times N}$ has orthonormal columns and $\Lambda\in\R^{N\times N}$ is diagonal, then
\begin{equation}
\Big( F_s^TWF_s + \mathrm{I}_n \Big)^{-1} = \mathrm{I}_n - Q\tilde{\Lambda}Q^T
\end{equation}
and $\tilde{\Lambda}_i=\lambda_i/(1 + \lambda_i)$. It only requires matrix vector products with a cost of $\mO(n\log(n))$ to compute \eqref{interp_eq:col_F} in the construction of Chebyshev approximation,
\begin{equation}
\Big(F_s^TWF_s + \mathrm{I}_n\big)^{-1}\tilde{f}_i = \tilde{f}_i -  Q\tilde{\Lambda}Q^T\tilde{f}_i.
\end{equation}

\subsection{Solve QP with Interior Point Method}

In this subsection, we focus on solving the QP \eqref{eq:quadprog} with an interior point method, following the procedure in \cite[\S 16.6]{Nocedal_book}. First we introduce slack variable $s\succeq0$ and write down the KKT condition for \eqref{eq:quadprog}:
\[
\begin{cases}
&H\cdot p^k + g - A^T\lambda = 0\\
&A\cdot p^k - s - b = 0\\
&s_i\cdot\lambda_i = 0, \quad i=1,2,...,2n+1\\
&(s,\ \lambda)\succeq0.
\end{cases}
\]
We then define a complementarity measure $\mu = s^T\cdot \lambda/ (2n+1)$, and solve a linear system:
\begin{equation}\label{ip_biglinsys}
\begin{pmatrix}
H & 0 & -A^T\\
A & -I & 0\\
0 & \Lambda & S
\end{pmatrix}
\begin{pmatrix}
\Delta p^k\\ \Delta s\\ \Delta \lambda
\end{pmatrix}
=
\begin{pmatrix}
-r_d\\
-r_p\\
-\Lambda\cdot S\bm{1} + \sigma\cdot\mu\bm{1}
\end{pmatrix}
\end{equation}
where $$r_d = H\cdot p^k-A^T\lambda+g, \quad r_p=A\cdot p^k - s - b$$
and 
$$\Lambda = \diag(\lambda_1,..,\lambda_{2n+1}), \quad S = \diag(s_1,..,s_{2n+1}), \quad \bm{1} = (1,1,..,)^T.$$ A more compact ``normal equation'' form of the system \eqref{ip_biglinsys} is
\begin{equation}\label{ip_linsys}
\Big( H+A^TS^{-1}\Lambda A\Big)\Delta p^k = -r_d + A^TS^{-1}\Lambda\big( -r_p-s+\sigma\mu\Lambda^{-1}\bm{1}\big)
\end{equation}
Next we solve the linear system \eqref{ip_linsys}. Note that once $\Delta p^k$ is known, $\Delta s$ and $\Delta\lambda$ can be derived easily. Let $S^{-1}\Lambda=\diag(d_1, d_2,..,d_{2n+1})$, we have from \eqref{eq:quadprog} that
\begin{equation}
A^TS^{-1}\Lambda A = D + d_{2n+1}\cdot\bm{1}\cdot\bm{1}^T
\end{equation}
where $D=\diag(d_1+d_{n+1}, d_2+d_{n+2},..,d_n+d_{2n})$. We apply Sherman-Morrison formula to calculate $(H + D + d_{2n+1}\cdot\bm{1}\cdot\bm{1}^T)^{-1}$. Since $H\approx H_s=C^T\tilde{H}C$ and it is less expensive to compute $\tilde{H}^{-1}$, we have
\begin{equation}\label{ip_inv1}
\Big(H_s + D\Big)^{-1} = \Big(C^T\tilde{H}C + D\Big)^{-1} = D^{-1} - D^{-1}C^T\Big( \tilde{H}^{-1} + CD^{-1}C^T \Big)^{-1}CD^{-1}.
\end{equation}
Note $\tilde{H}$ is positive semi-definite, but not necessarily positive definite. We use a truncated eigenvalue decomposition of $\tilde{H}$ and only considers eigenvalues above a threshold, and then compute the Moore-Penrose inverse. Next let $X:=H_s + D$, and we apply \eqref{ip_inv1} to get
\begin{equation}
(H + D + d_{2n+1}\cdot\bm{1}\cdot\bm{1}^T)^{-1}=\Big(X + d_{2n+1}\bm{1}\cdot\bm{1}^T \Big)^{-1} = X^{-1}  - X^{-1} \bm{1}\bm{1}^TX^{-1}/\Big( \bm{1}^TX^{-1}\bm{1} + d_{2n+1}^{-1} \Big).
\end{equation}
To summarize, we only need $\tilde{H}^{-1}$ and matrix vector products to solve for $\Delta p^k$ in \eqref{ip_linsys}. They are of complexity $\mathcal{O}(\log^3(n))$ and $\mathcal{O}(n\log(n))$ respectively, and both are affordable to compute. We implement Algorithm 16.4 in \cite{Nocedal_book} to solve \eqref{eq:quadprog}, and because the number of iterations with increasing variable dimensions is usually stable for interior point algorithms, our SQP algorithm has an overall complexity of $O(n\log^2(n))$.

\section{Error Analysis - Convergence in Integrality Gap}\label{algo_sec:error}\label{sec:error}
In this section, we analyze the accuracy of our algorithm from the interpolation-based approximations. Specifically, let $w^N$ be the solution from SQP (Algorithm \ref{sqp_algo:sqp}), $w^{N, int}$ be the integer solution constructed from SUR (see Appendix \ref{app:sur}), and $w^n$ be the true solution to the relaxed program \eqref{algo_eq:doe2}, we show that with the full $F$ matrix, as $n\to\infty$
\begin{equation}
\Big|\phi\big(\Gamma_{post}(w^n)\big) - \phi\big(\Gamma_{post}(w^{N, int})\big)\Big| \to 0.
\end{equation}
Note $N=\mO(\log(n))$ also increases to infinity. It tells us if we solve the optimization with the low-rank approximation matrix $F_s$, the objective value of the SUR integer solution converges to the true minimum. We address this problem in the following two subsections.
 
\subsection{Connection Between Two Optimization Problems}
Algorithm \ref{sqp_algo:sqp} returns a solution to the following convex program:
\begin{equation}\label{eq:doe_lowrank}
\begin{array}{rl}
\mbox{min}\quad &\phi_s(\Gamma_{post}(w))\\
\mbox{s.t.}\quad  &0\le w_i\le 1,\ \sum_{i=1}^n w_i\le n_0.
\end{array}
\end{equation}
where $\phi_s(\Gamma_{post}(w)) = \phi\big( (F_s^TWF_s + \mathrm{I}_n)^{-1}\big)$. The program differs from the original relaxation \eqref{algo_eq:doe2} only in $F$, and for simplicity, we use the abbreviation $\phi_s(w)$ for $\phi_s(\Gamma_{post}(w))$, and $\phi(w)$ for $\phi(\Gamma_{post}(w))$.

\begin{claim}\label{error_claim:phi1}
Let $w^N$, $w^n$ be the solution to \eqref{eq:doe_lowrank} and \eqref{algo_eq:doe2} respectively. If $|\phi(w)-\phi_s(w)| < \epsilon$ for any $w\in\mathbb{R}^n$, then
\begin{equation}
|\phi(w^N) - \phi(w^n)| < 2\epsilon.
\end{equation}
\end{claim}
In other words, if $\phi_s$ is close to $\phi$ for any $w$, then the objective value of $w^N$ is close to the true minimum $ \phi(w^n)$.
\begin{proof}
Because $w^N$ and $w^n$ minimizes $\phi_s(w)$ and $\phi(w)$ respectively, we have
\begin{equation}\label{ineq: obj_1}
\phi_s(w^N)\le\phi_s(w^n),\quad \phi(w^n)\le\phi(w^N).
\end{equation}
From the assumption $|\phi(w)-\phi_s(w)| < \epsilon$, we know 
\begin{equation}
|\phi(w^N) - \phi_s(w^N)| < \epsilon, \quad |\phi(w^n) - \phi_s(w^n)| < \epsilon,
\end{equation}
and together with \eqref{ineq: obj_1}, we get
\begin{equation}\label{ineq: obj_2}
\phi(w^N)\le\phi_s(w^N)+\epsilon\le\phi_s(w^n)+\epsilon<\phi(w^n) + 2\epsilon.
\end{equation}
The result follows directly from \eqref{ineq: obj_1} and \eqref{ineq: obj_2}.
\end{proof}

Next we show $|\phi(w)-\phi_s(w)|$ is small for any $w\in\mathbb{R}^n$. Because $\frac{1}{\Delta y}F_s(i,j)$ is an approximation of $\frac{1}{\Delta y}F(i,j)$ which equals $f(x_i,y_j)$, their difference is determined by the error in Chebyshev approximtaion (see \cite{greenbaum2012numerical}), and we quantify $|\phi(w)-\phi_s(w)|$ now. 

We use the notation $\|X\|=\|X\|_F$ (Frobenius norm) for any matrix $X$ from now on.

\begin{claim}\label{algo_claim:obj1}
If $\frac{1}{\Delta y}|F(i,j) - F_s(i,j)| < \epsilon$, then for any $w\in\mathbb{R}^n$,
\begin{equation}
|\phi(w) - \phi_s(w)| \le C\cdot N\cdot\epsilon,
\end{equation}
where $C$ is a positive constant independent of $n$ and $N$.
\end{claim}

Note $F_s$ is defined in \eqref{interp_eq:F_s} with $N$ interpolation points, and $\epsilon$ represents the interpolation error which we will discuss in the next subsection.
\begin{proof}
Because $|F(i,j) - F_s(i,j)|<\epsilon\Delta_y$ for $i,j = 1,2,..,n$, we have 
\begin{equation}\label{algo_eq:FFS}
\|F-F_s\| < \sqrt{\sum_{i=1}^n\sum_{j=1}^n \epsilon^2\Delta^2y} = n\Delta y\cdot \epsilon = \mu(\indom)\cdot\epsilon.
\end{equation}
Similarly because $|\frac{1}{\Delta y}F(i, j)|=|f(x_i,y_j)|\le\max |f(x, y)|$,  
\begin{equation}\label{algo_eq:boundF}
\|F\| = \sqrt{\sum_{i=1}^n\sum_{j=1}^n F(i, j)^2(\Delta y)^2} \le n\Delta y\cdot\max |f(x,y)| = \mu(\indom)\cdot\max |f(x,y)|.
\end{equation}
Moreover, we can show $\|F_s\|$ is also bounded
\begin{equation}
\|F_s\| = \|F_s+F-F\|\le \|F\|+\|F-F_s\| \le \mu(\indom)\cdot\max |f(x,y)| + \mu(\indom)\cdot\epsilon.
\end{equation}
When $\epsilon$ is small (e.g. $\epsilon < \max |f(x,y)|$), we get
\begin{equation}\label{algo_eq:boundFS}
\|F_s\| \le 2\mu(\indom)\cdot\max |f(x,y)|.
\end{equation}
Because $W$ is a diagonal matrix with each component between 0 and 1, the matrix product $WF$ results in multiplying the $i^{th}$ row of $F$ by $w_i$ and thus $\|WF\|\le\|F\|$. For similar reasons, we have $\|FW\|\le\|F\|$ and get
\begingroup
\addtolength{\jot}{0.5em}
\begin{align}
\big\| F^TWF - F_s^TWF_s \big\| &\le \big\| F^TW(F-F_s)\big\| + \big\|(F-F_s)^TWF_s \big\| \\
&\le \| F^TW\|\cdot \|(F-F_s)\| + \|F-F_s\|\cdot \|WF_s\| \\
&\le \| F\|\cdot \|(F-F_s)\| + \|F-F_s\|\cdot \|F_s\| \\
&< c\epsilon
\end{align}
\endgroup
where the postive constant $c=3\cdot\mu^2(\indom)\cdot\max |f(x,y)|$ from \eqref{algo_eq:FFS}, \eqref{algo_eq:boundF} and \eqref{algo_eq:boundFS}. Let
\begin{equation}
\lambda^n_1\ge\lambda^n_2\ge\cdots\ge\lambda^n_n, \quad\lambda^{n,s}_1\ge\lambda^{n,s}_2\ge\cdots\ge\lambda^{n,s}_n
\end{equation}
be the eigenvalues of $F^TWF$ and $F_s^TWF_s$ respectively. In \cite{sur_oed}, it has been proved
\begin{equation}\label{ineq: FFS_entry}
| \lambda^n_i - \lambda^{n,s}_i | < \|F^TWF - F_s^TWF_s \| < c\epsilon.
\end{equation}
Because the rank of $F_s$ is at most $N$, we have $\lambda^{n,s}_{N+1}=..=\lambda^{n,s}_n=0$. In the trace case, 
\begin{align*}
|\phi(w) - \phi_s(w)| &= \big|\sum_{i=1}^n\frac{1}{1 + \lambda^n_i} - \sum_{i=1}^n\frac{1}{1 + \lambda^{n,s}_i} \big| \\
& \le \big| \sum_{i=1}^N \big(\frac{1}{1+\lambda^n_i} - \frac{1}{1+\lambda^{n,s}_i} \big) \big| + \big| \sum_{i=N+1}^n \big(\frac{1}{1+\lambda^n_i} - \frac{1}{1+\lambda^{n,s}_i} \big) \big| \\
& = \sum_{i=1}^N\frac{|\lambda^n_i-\lambda^{n,s}_i|}{(1+\lambda^n_i)(1+\lambda^{n,s}_i)} + \sum_{i=N+1}^N\frac{\lambda^n_i}{1+\lambda^n_i} \\
& \le \sum_{i=1}^N|\lambda^n_i-\lambda^{n,s}_i| + \sum_{i=N+1}^n\lambda^n_i.
\end{align*}

We control the two terms separately. The first term $\sum_{i=1}^N|\lambda^n_i-\lambda^{n,s}_i| $ is bounded by $c\cdot N\epsilon$ from \eqref{ineq: FFS_entry}, and for the other, we first notice that 
\begin{align*}
 \big|tr(F^TWF) - tr(F_s^TWF_s)\big| &= \big|\Delta_y^2\sum_{i,j} w_if^2(x_i,y_j) - \Delta_y^2\sum_{i,j}w_if_s^2(x_i,y_j)\big| \\
& = \big| \Delta_y^2\sum_{i,j} w_i\big(f(x_i,y_j)+f_s(x_i,y_j)\big)\big(f(x_i,y_j)-f_s(x_i,y_j)\big) \big| \\
& \le \Delta_y^2\sum_{i,j} (2c_f\cdot\epsilon) \\
& = 2(n\Delta y)^2c_f\cdot\epsilon=: \tilde{c}\epsilon
\end{align*}
where $c_f$ is the uniform bound for both $|f(x_i,y_j)|$ and $|f_s(x_i,y_j)|$. The constant $\tilde{c}$ depends on $c_F$ and $\mu(\indom)$. The last by two step is due to the claim assumption and $f(x_i,y_j)-f_s(x_i,y_j) = \frac{1}{\Delta y}(F(i,j) - F_s(i,j))$. Because the trace function can be expressed as the sum of eigenvalues, we have
\begin{align*}
\big|tr(F^TWF) - tr(F_s^TWF_s)\big|  &= \big| \sum_{i=1}^n\lambda^n_i - \sum_{i=1}^N\lambda^{n,s}_i| \\
&= \big| \sum_{i=1}^N (\lambda^n_i-\lambda^{n,s}_i) + \sum_{i>N}\lambda^{n}_i \big| \\
&\ge -\sum_{i=1}^N |\lambda^n_i-\lambda^{n,s}_i| + \sum_{i>N}\lambda^n_i
\end{align*}
which implies
\begin{equation}
\sum_{i>N}\lambda^n_s \le \big|tr(F^TWF) - tr(F_s^TWF_s)\big| + \sum_{i=1}^N |\lambda^n_i-\lambda^{n,s}_i|  \le \tilde{c}\epsilon + c\cdot N\epsilon.
\end{equation}
Therefore,
\begin{equation}\label{algo_eq:bound_phi}
|\phi(w) - \phi_s(w)| \le (\tilde{c} + 2cN) \epsilon,
\end{equation}
where $\tilde{c}$ and $c$ are constants free of $n$ and $N$. When $N>\tilde{c}$, we get for any $w\in\mathbb{R}^n$,
\begin{equation}
|\phi(w) - \phi_s(w)| \le (2c+1)N\cdot\epsilon =: C\cdot N\cdot\epsilon.
\end{equation}
which completes the proof.
\end{proof}

\subsection{Chebyshev Interpolation Error}
Claim \ref{algo_claim:obj1} suggests that to get accurate approximation, we should make $N\epsilon$ small where $\epsilon$ is the error in Chebyshev polynomial approximation, which depends on the number of interpolation points $N$ and the smoothness of $f(x,y)$. Classical theory on Chebyshev interpolation error is well developed (e.g. see \cite{greenbaum2012numerical}) and gives the following result: Let $f$ be a continuous function on $[-1,1]$, $h_n$ be its degree $n$ polynomial interpolant at the Chebyshev nodes, $\varepsilon = \|f-h_n\|_{\infty}$, 
\begin{itemize}
\item if $f$ has a $k^{th}$ derivative of bounded variation, then $\varepsilon = \mathcal{O}(N^{-k})$;
\item if $f$ is analytical in a neighborhood of $[-1,1]$, then $\varepsilon = \mathcal{O}(\rho^N)$ for some $0<\rho<1$.
\end{itemize}

We are now ready to bound $\frac{1}{\Delta y}|F(i,j) - F_s(i,j)|$ to satisfy the condition in Claim \ref{algo_claim:obj1}. Note that $\frac{1}{\Delta y}F(i,j) = f(x_i, y_j)$ and 
\begin{equation}
\frac{1}{\Delta y}F_s(i,j) = \sum_{p=1}^N\sum_{q=1}^N l_p(x_i)l_q(y_j)f(\tilde{x}_p, \tilde{y}_q)
\end{equation}
where
\begin{equation}
l_p(x) = \prod_{\substack{k = 1\\ k\neq p}}^N\frac{x-\tilde{x}_k}{\tilde{x}_p-\tilde{x}_k}, \quad l_q(y) = \prod_{\substack{k = 1\\ k\neq q}}^N\frac{y-\tilde{y}_k}{\tilde{y}_q-\tilde{y}_k}
\end{equation}
and $\{\tilde{x}_p\}_{i=1}^N$ and $\{\tilde{y}_q\}_{i=1}^N$ are interpolation points in $\outdom$ and $\indom$ respectively. In one dimension, if for $\forall x\in\outdom, \forall y\in\indom$,
\begin{equation}
\Big| f(x,y) - \sum_{p=1}^Nl_p(x)f(\tilde{x}_p, y) \Big| \le \epsilon_0, \quad \Big| f(x,y) - \sum_{q=1}^Nl_q(y)f(x, \tilde{y}_q) \Big| \le \epsilon_0,
\end{equation}
then 
\begin{align}
& \frac{1}{\Delta y}|F(i,j) - F_s(i,j)| \\
& = \Big| f(x_i, y_j) - \sum_{p=1}^N\sum_{q=1}^N l_p(x_i)l_q(y_j)f(\tilde{x}_p, \tilde{y}_q) \Big| \\
& = \Big| f(x_i, y_j)  - \sum_{p=1}^Nl_p(x_i)f(\tilde{x}_p, y_j) +  \sum_{p=1}^Nl_p(x_i) \Big( f(\tilde{x}_p, y_j) - \sum_{q=1}^N l_q(y_j)f(\tilde{x}_p, \tilde{y}_q) \Big) \Big| \\
&\le \Big| f(x_i, y_j)  - \sum_{p=1}^Nl_p(x_i)f(\tilde{x}_p, y_j)  \Big| + \sum_{p=1}^N|l_p(x_i)|\Big| f(\tilde{x}_p, y_j) - \sum_{q=1}^N l_q(y_j)f(\tilde{x}_p, \tilde{y}_q) \Big| \\
&\le \epsilon_0 + \epsilon_0\sum_{p=1}^N |l_p(x_i)|\\
&\le \epsilon_0 + \Lambda_N*\epsilon_0 \label{algo_eq:bound_Fij}
\end{align}
where $\Lambda_N$ is called the Lebesgue constant and it is the opeator norm of Lagragian interpolation polynomial projection at Chebyshev nodes. It is known (see \cite{LebCons}) that
\begin{equation}
\frac{2}{\pi}\log(N) + a < \Lambda_N < \frac{2}{\pi}\log(N) + 1, \quad a = 0.9625....
\end{equation}
\begin{itemize}
\item If $f(x,y)$ is $k^{th}$ order continuously differentiable, in the one-dimensional case, we have
\begin{align*}
\epsilon_0 = \mathcal{O}(N^{-k}) \quad &\xRightarrow[]{\eqref{algo_eq:bound_Fij}}\quad \frac{1}{\Delta y}|F(i,j) - F_s(i,j)| = \mathcal{O}(N^{-k}\log(N)) \\ &\xRightarrow[]{\eqref{algo_eq:bound_phi}}\quad |\phi(w) - \phi_s(w)| = \mathcal{O}(N^{1-k}\log(N)).
\end{align*}
We conclude that when $f(x,y)$ is at least 2nd order differentiable, $ |\phi(w) - \phi_s(w)| $ will diminish as $n\to\infty$. The decay gets slower in multiple dimensions intuitively because if $\indom, \outdom\subset R^d$, there are only $N^{1/d}$ interpolation points on each dimension. For the purpose of easy presentation, let's still assume there are $N$ interpolation points on each dimension. To derive an error bound, we define the Chebyshev interpolation operator $\mathcal{I}_N$ that maps a function $h\in\mathbb{C}([-1,1])$ to a degree-$N$ polynomial:
\begin{equation}
\mathcal{I}_N h(x) = \sum_{p=1}^Nl_p(x)h(\tilde{x}_p).
\end{equation}
Since $\Lambda_N$ is the operator norm, we have $\|\mathcal{I}_N h\|_\infty \le \Lambda_N\|h\|_\infty$. For $x,y\in\mathbb{R}^d$, we now define the double $d$-th order tensor product interpolation operator:
\begin{equation}
\mathcal{I}^{d,d}_Nf(x, y) = \mathcal{I}_{N,x}^1\times\cdots\times\mathcal{I}_{N,x}^d\times\mathcal{I}_{N,y}^1\times\cdots\times\mathcal{I}_{N,y}^df(x,y),
\end{equation}
where $\mathcal{I}_{N,x}^i$ denotes the single interpolation operator along the $i$-th dimension in the output domain $\outdom$, and similarly $\mathcal{I}_{N,y}^j$ is the single interpolation operator along the $j$-th dimension in $\indom$.
\begingroup
\addtolength{\jot}{0.5em}
\begin{align}
|f(x,y) - \mathcal{I}^{d,d}_Nf(x,y)| =&\ |f(x,y) -\mathcal{I}_{N,x}^1f(x,y) + \mathcal{I}_{N,x}^1f(x,y) - \mathcal{I}_Nf(x,y)|\\
\le&\ |f(x,y) -\mathcal{I}_{N,x}^1f(x,y)| +  |\mathcal{I}_{N,x}^1f(x,y) - \mathcal{I}_Nf(x,y)|\\
\le&\ |f(x,y) -\mathcal{I}_{N,x}^1f(x,y)| +  |\mathcal{I}_{N,x}^1f(x,y) - \mathcal{I}_{N,x}^1\mathcal{I}_{N,x}^2f(x,y)| \\
&+ |\mathcal{I}_{N,x}^1\mathcal{I}_{N,x}^2f(x,y) - \mathcal{I}_Nf(x,y)| \\
\le&\ |f(x,y) -\mathcal{I}_{N,x}^1f(x,y)| +  \Lambda_N|f(x,y) - \mathcal{I}_{N,x}^2f(x,y)| \\
&+ \cdots + \Lambda_{N}^{2d-1}|f(x,y) - \mathcal{I}_{N,y}^df(x,y)| \\
\le&\ \epsilon_0(1 + \Lambda_N + \Lambda_N^2 + \cdots + \Lambda_N^{2d-1})\\
=&\ \epsilon_0\frac{\Lambda_N^{2d} - 1}{\Lambda_N - 1} 
\le \epsilon_0\cdot\Lambda_N^{2d} \label{algo_eq:bound_tensor}
\end{align}
\endgroup
for any $\Lambda_N > 2$. We implicitly assume here $\epsilon_0$ is the uniform bound of the interpolation error on any single dimension in $\indom$ and $\outdom$. Now we go back to the case where there are $N$ interpolation points in total, so there are $N^{1/d}$ interpolations on each side and if $f(x,y)$ is $k$-th order continuously differentiable,
\begin{align*}
\epsilon_0 = \mathcal{O}(N^{-k/d}) \quad &\xRightarrow[]{\eqref{algo_eq:bound_tensor}}\quad \frac{1}{\Delta y}|F(i,j) - F_s(i,j)| = \mathcal{O}(N^{-k/d}\log^{2d}(N)) \\ &\xRightarrow[]{\eqref{algo_eq:bound_phi}}\quad |\phi(w) - \phi_s(w)| = \mathcal{O}(N^{1-k/d}\log^{2d}(N)).
\end{align*}
In order to guarantee the convergence of $|\phi(w) - \phi_s(w)|$ to zero, $f(x,y)$ should be at least $(d+1)$-th continuously differentiable.
\item If $f(x,y)$ is analytical, then $\epsilon_0 = \mathcal{O}(\rho^{N^{1/d}})$ for some $0<\rho<1$ , then
\begin{equation}\label{error_eq:error_phi_inf}
|\phi(w) - \phi_s(w)| = \mathcal{O}\big(N\log^{2d}(N)\rho^{N^{1/d}}\big).
\end{equation}
which converges to zero for any dimension $d$.
\end{itemize}
\begin{theorem}
If $f(x,y)$ is an analytical function on $\outdom\times\indom$, then as $n\to\infty$,
\begin{equation}\label{error_eq:gap_full}
\Big|\phi\big(\Gamma_{post}(w^n)\big) - \phi\big(\Gamma_{post}(w^{N, int})\big)\Big| \to 0.
\end{equation}
\end{theorem}

\begin{proof}
From \eqref{error_eq:error_phi_inf} we have $|\phi(w) - \phi_s(w)|\to0$ for any $w\in\R^n$, and then by Claim \ref{error_claim:phi1},
\begin{equation}\label{error_eq:relax_erro}
\Big|\phi\big(\Gamma_{post}(w^n)\big) - \phi\big(\Gamma_{post}(w^{N})\big)\Big| \to 0.
\end{equation}
From \citep[Theorem 3]{sur_oed} on the zero convergence of integrality gap, we know
\begin{equation}
\Big|\phi\big(\Gamma_{post}(w^N)\big) - \phi\big(\Gamma_{post}(w^{N, int})\big)\Big| \to 0.
\end{equation}
Combine the above inequalities, we have the convergence of SQP integer solution
\begin{align}
&\ \Big|\phi\big(\Gamma_{post}(w^n)\big) - \phi\big(\Gamma_{post}(w^{N, int})\big)\Big|  \\
\le&\ \Big|\phi\big(\Gamma_{post}(w^n)\big) - \phi\big(\Gamma_{post}(w^{N})\big)\Big| + \Big|\phi\big(\Gamma_{post}(w^N)\big) - \phi\big(\Gamma_{post}(w^{N, int})\big)\Big| \\
\to&\ 0
\end{align}
and the proof is complete.
\end{proof}

In \S\ref{sec:interp}, we choose $N =c\log(n)$ to achieve the computational complexity $\mathcal{O}(n\log^2(n))$, but an important question is how to choose the constant $c$. One practical suggestion is to solve for problems with moderate sizes and get the exact solution (true minimum), and then adjust the constant $c$ by doubling it until the error in \eqref{error_eq:relax_erro} falls below a preassigned threshold. The trade off in the choice of $c$ should be clear: when $c$ is larger, the approximation is more accurate, but it is more computationally expensive.

\section{Temporal and Two-Dimensional LIDAR Problem}\label{sec:lidar}
\label{algo_sec:lidar}
In this section, we apply the sequential quadratic programming in \S\ref{sec:sqp} to solve a Bayesian inverse problem driven by partial differential equations. Specifically, our goal is to infer the initial condition of an advection-diffusion equation on a spatial and temporal domain, where the observable can be expressed as a truncated sum of integral equations so that all the convergence results would apply.

\subsection{Extend Convergence Results to Space-time Models}
Because we are adding an extra time domain, theorems need to be extended for time-dependent measurements. For this extension, we require that the measurements be taken at a fixed frequency.

\subsubsection{Parameter-to-observable Map}
Consider a compact domain $V$ in $\mathbb{R}^P$ and a time interval $[0,T]$. Suppose the measurement without noise has the integral-equation form: for $x\in\outdom$
\begin{equation}\label{eq:meas}
u(x,t)=\int_{\indom} f(x,y,t)u_0(y)\,\mathrm{d}y.
\end{equation}
In our example, $u(x,t)$ is the solution to a partial differential equation describing a dynamical system, and $f(x,y,t)$ is derived from solving the equation. We discretize the integral equation \eqref{eq:meas} as before, construct a matrix $F$ from $f(x,y,t)$, and divide the domain $\outdom$ ($\indom$ and $[0,T]$) into $n_x$ ($n_y$ and $n_t$) equally spaced intervals ($\Delta x =\mu(\outdom)/n_x, \Delta y=\mu(\outdom)/n_y, \Delta t=T/n_t$). Then,
$
\hat{u}=F\hat{u}_0\in\mathbb{R}^{n_xn_t\times1}, F\in\mathbb{R}^{n_xn_t\times n_y},
$
and
$$
\hat{u} = \big(u(x_1, t_1), u(x_1,t_2), ...,u(x_1,t_{n_t}), u(x_2,t_1), ...,u(x_2,t_{n_t}), ..., u(x_{n_x},t_1),u(x_{n_x}, t_{n_t}) \big)^T
$$
$$
F=
\begin{pmatrix}
f(x_1,y_1,t_1) & f(x_1,y_2,t_1) & \cdots & f(x_1,y_{n_y},t_1)\\
f(x_1,y_1,t_2) & f(x_1,y_2,t_2) & \cdots & f(x_1,y_{n_y},t_2)\\
\vdots & \vdots & & \vdots\\
f(x_1,y_1,t_{n_t}) & f(x_1,y_2,t_{n_t}) & \cdots & f(x_1,y_{n_y},t_{n_t})\\
f(x_2,y_1,t_1) & f(x_2,y_2,t_1) & \cdots & f(x_2,y_{n_y},t_1)\\
\vdots & \vdots & & \vdots\\
f(x_{n_x},y_1,t_{n_t}) & f(x_{n_x},y_2,t_{n_t}) & \cdots & f(x_{n_x},y_{n_y},t_{n_t})
\end{pmatrix}\Delta y
$$
and $\hat{u}_0\in \mathbb{R}^{n_y}$ is a discretization of $u_0(x)$ with $\hat{u}_{0,j}=u_0(y_j)$ ($j=1,2,...,n_y$). To figure out the $(i,j)^{th}$ entry of $F$, let $i\!=\!(i_1-1)n_t+i_2$ ($i_1\in\{1,2,...,n_x\}, i_2\in\{1,2,...,n_t\}$) and $j=1,2,...,n_y$, we have
\begin{equation}
F(i,j)=f(x_{i_1},y_j, t_{i_2})\Delta y.
\end{equation}
If $\indom = \outdom$, we use the same discretization, i.e. $x_j=y_j$ ($j=1,2...,n_x$) and $n_x=n_y$. 

\begin{remark}
$f(x,y,t)$ in \eqref{eq:meas} is not always continuous as a solution to PDEs, for example, $f(x,y,t)$ in a one-wave system is a delta function $\delta(x\!-\!at,y)$ where $a$ is the wave speed. 
\end{remark}

\subsubsection{Convexity of the Objective Function}
In our Bayesian framework, the posterior covariance matrix is given by
\begin{equation}
\postcov=\Big(F^TW^{1/2}\Gamma_{noise}^{-1}W^{1/2}F+\prcov^{-1}\Big)^{-1}.
\end{equation}
where $\Gamma_{noise}$ is the noise covariance matrix among measurements. We assume the measurement noise is only correlated in time, not in space, so $\Gamma_{noise}$ is block diagonal and the number of blocks equals the number of discrete points on the spacial domain $\outdom$.

\begin{lemma}
$tr(\postcov)$ and $\log\det(\postcov)$ are convex functions in the weight vector $w$.
\end{lemma}

\begin{proof}
We construct the matrix $W$ from the weight vector $w\!=\!(w_0,w_1,..,w_{n_x})$ as follows:
$$
W=\diag\{w_1,w_1,...,w_1,w_2, w_2,..,w_2,..,w_{n_x},w_{n_x},...,w_{n_x}\}\in\mathbb{R}^{n_xn_t \times n_xn_t}.
$$
Note that $\Gamma_{noise}^{-1}$ is also block diagonal, and we write $\ncov^{-1}$, $F$ and $W$ as
$$
\ncov^{-1}=
\begin{pmatrix}
P_1 & \cdots & \cdots & \cdots\\
\cdots & P_2 & \cdots & \cdots\\
\vdots & \vdots & \ddots & \vdots\\
\cdots & \cdots & \cdots & P_{n_x}
\end{pmatrix}
\qquad
F=
\begin{pmatrix}
F_1\\
F_2\\
\vdots\\
F_{n_x}
\end{pmatrix}
\qquad
W=
\begin{pmatrix}
w_1I_{n_t} & \cdots & \cdots & \cdots\\
\cdots & w_2I_{n_t} & \cdots & \cdots\\
\vdots & \vdots & \ddots & \vdots\\
\cdots & \cdots & \cdots & w_{n_x}I_{n_t}
\end{pmatrix}
$$
where $P_k\equiv P\in\mathbb{R}^{n_t\times n_t}$ is the precision matrix for each location and $F_k\in\mathbb{R}^{n_t\times n_y}$. 
$$
\postcov=\Big(\sum_{k=1}^{n_x}w_kF_k^TP_kF_k+\prcov^{-1}\Big)^{-1}
$$
The desired results follow because $tr(X^{-1})$ and $\log\det(X^{-1})$ are both convex in $X$ (see \citep[Exercise 3.26]{boyd}), and $X$ is linear in $w$.
\end{proof}

\subsubsection{Extend the Convergence Theory}
The inverse of $\postcov$ is 
\begin{equation}
\postcov^{-1}=\sum_{k=1}^{n_x}w_kF_k^TPF_k+\prcov^{-1},
\end{equation}
and if we denote $f_{i,j,s} = f(x_i, y_j, t_s)$, the matrix $F_k^TPF_k$ can be written as 
\[
\begin{pmatrix}
\sum_{s_1,s_2=1}^{n_t}f_{k,1,s_1}P_{s_1,s_2}f_{k,1,s_2} & \cdots & \sum_{s_1,s_2=1}^{n_t}f_{k,1,s_1}P_{s_1,s_2}f_{k,n_y,s_2}\\
\vdots & \ddots & \vdots\\
\sum_{s_1,s_2=1}^{n_t}f_{k,n_y,s_1}P_{s_1,s_2}f_{k,1,s_2} & \cdots & \sum_{s_1,s_2=1}^{n_t}f_{k,n_y,s_1}P_{s_1,s_2}f_{k,n_y,s_2}
\end{pmatrix}(\Delta y)^2.
\]
Therefore, the $(i,j)^{th}$ entry in $\postcov^{-1}$ is
$$
\postcov^{-1}(i,j) = \Delta y\sum_{k=1}^{n_x}\sum_{s_1=1}^{n_t}\sum_{s_2=1}^{n_t}w_kf(x_k, t_{s_1},y_i)P_{s_1,s_2}f(x_k,t_{s_2},y_j)\Delta x.
$$
If measurements are collected every few seconds or minutes within a time range, i.e. $n_t$ is a fixed integer, then for any precision matrix $P$, we are in the same setting as \citep{sur_oed}, and all the convergence proofs can be extended trivially.

\subsection{Two-dimensional Advection-diffusion Equation with External Source}
The advection-diffusion equation is a combination of diffusion and advection equations, and we first solve the diffusion equation (or heat equation), which lays the foundation to solving the advection-diffusion equation. We add an external force to the equation to keep the system from entering a stationary state. Later we shall see this external force has no effect on the design when the goal is to infer the initial condition.

\subsubsection{Two-dimensional Heat Equation with External Source}
Consider the heat equation on a two-dimensional domain $[-1,1]\times[-1, 1]$ with homogeneous Dirichlet boundary conditions, i.e. $u(x,y,t) = 0 \mbox{ for }(x, y)\mbox{ on the boundary}$, and an external force $f(x,y,t)$:
\begin{equation}\label{lidar_eq:2d_heat_ext}
u_t - \mu\nabla u = u_t - \mu \big(\frac{\partial^2 u}{\partial x^2} + \frac{\partial^2 u}{\partial y^2} \big) = f(x, y, t)
\end{equation}
The initial condition $u(x, y, 0) = u_0(x, y)$ is the parameter of interest. Using a variant of separation of variables, we assume the solution $u(x,y,t)$ has the following form
\begin{equation}
u(x, y,t) = \sum_k T_k(t)X_k(x,y)
\end{equation}
and then apply the tensor product of Fourier basis $\{\sin(\frac{k\pi x}{2}), \cos(\frac{k\pi x}{2})\}_{k\ge0}$ in one dimension, and rewrite $u(x,y,t)$ as
\begin{align}
u(x, y, t) = & \sum_{k_1\ge0}\sum_{k_2\ge0} \Big\{ T_{k_1, k_2}^{(1)}(t)\sin(\frac{k_1\pi x}{2})\sin(\frac{k_2\pi y}{2}) + T_{k_1, k_2}^{(2)}(t)\sin(\frac{k_1\pi x}{2})\cos(\frac{k_2\pi y}{2}) \\ & + T_{k_1, k_2}^{(3)}(t)\cos(\frac{k_1\pi x}{2})\sin(\frac{k_2\pi y}{2}) + T_{k_1, k_2}^{(4)}(t)\cos(\frac{k_1\pi x}{2})\cos(\frac{k_2\pi y}{2}) \Big\}.
\end{align}
We treat the four terms separately, and exemplify it with $\sin(\frac{k_1\pi x}{2})\sin(\frac{k_2\pi y}{2})$. Results for the other three terms can be derived similarly. Let 
\begin{equation}
u^{(1)}(x, y, t) = \sum_{k_1\ge0}\sum_{k_2\ge0} T_{k_1,k_2}^{(1)}(t)\sin(\frac{k_1\pi x}{2})\sin(\frac{k_2\pi y}{2}),
\end{equation}
and apply the diffusion dynamics to $u^{(1)}(x, y, t)$
\begin{equation}
\frac{\partial u^{(1)}}{\partial t} - \mu (\frac{\partial^2 u^{(1)}}{\partial x^2}+ \frac{\partial^2 u^{(1)}}{\partial y^2})  = \sum_{k_1\ge0}\sum_{k_2\ge0} \Big(\frac{\partial T_{k_1,k_2}^{(1)}(t)}{\partial t} + \mu\frac{k_1^2 + k_2^2}{4}\pi^2T_{k_1,k_2}^{(1)}(t) \Big)\sin(\frac{k_1\pi x}{2})\sin(\frac{k_2\pi y}{2}).
\end{equation}
The external force also has a Fourier expansion, and we solve the following equation:
\begin{equation}\label{algo_eq:sub_pde}
\begin{cases}
\frac{\partial T_{k_1,k_2}^{(1)}(t)}{\partial t} + \mu\frac{k_1^2 + k_2^2}{4}\pi^2T_{k_1,k_2}^{(1)}(t) = f_{k_1,k_2}(t)\\
T_{k_1,k_2}^{(1)}(0) = c_{k_1,k_2}
\end{cases}
\end{equation}
where $f_{k_1,k_2}(t)$ and $c_{k_1,k_2}$ are the Fourier coefficients for the external force $f(x,y,t)$ and the initial condition $u_0(x,y)$ respectively, with respect to the basis $\sin(\frac{k_1\pi x}{2})\sin(\frac{k_2\pi y}{2})$:
\begin{align}
 f_{k_1,k_2}(t) & = \iint_{[-1,1]\times[-1,1]}f(x,y,t)\sin(\frac{k_1\pi x}{2})\sin(\frac{k_2\pi y}{2})\,\mathrm{d}x\mathrm{d}y,\\ c_{k_1,k_2} & = \iint_{[-1,1]\times[-1,1]}u_0(x,y)\sin(\frac{k_1\pi x}{2})\sin(\frac{k_2\pi y}{2})\,\mathrm{d}x\mathrm{d}y.
\end{align}
The solution to \eqref{algo_eq:sub_pde} is given by
\begin{equation}
T_{k_1,k_2}^{(1)}(t) = \exp\{-\mu\frac{k_1^2 + k_2^2}{4}\pi^2t\}c_{k_1,k_2} + \int_0^t\exp\{-\mu\frac{k_1^2 + k_2^2}{4}\pi^2(t-s)\}f_{k_1,k_2}(s)\,\mathrm{d}s.
\end{equation}
After working out the other three terms, we get the solution to \eqref{lidar_eq:2d_heat_ext}:
\begin{align}
u(t,x,y) = \sum_{k_1\ge0}\sum_{k_2\ge0} A_{k_1,k_2}(t)\sin(\frac{k_1\pi x}{2})\sin(\frac{k_2\pi y}{2}) + B_{k_1,k_2}(t)\sin(\frac{k_1\pi x}{2})\cos(\frac{k_2\pi y}{2}) \\+ C_{k_1,k_2}(t)\cos(\frac{k_1\pi x}{2})\sin(\frac{k_2\pi y}{2}) + D_{k_1,k_2}(t)\cos(\frac{k_1\pi x}{2})\cos(\frac{k_2\pi y}{2}) \label{lidar_eq:sol_heat_force}
\end{align}
where
\begin{align*}
A_{k_1,k_2}(t) = \exp\{-\mu\frac{k_1^2 + k_2^2}{4}\pi^2t\} \iint_{\indom} &u_0(x,y)\sin(\frac{k_1\pi x}{2})\sin(\frac{k_2\pi y}{2})\,\mathrm{d}x\mathrm{d}y \\&+ \int_0^t\exp\{-\mu\frac{k_1^2+k_2^2}{4}\pi^2(t-s)\}f_{k_1,k_2}^{(1)}(s)\,\mathrm{d}s  \\
B_{k_1,k_2}(t) = \exp\{-\mu\frac{k_1^2 + k_2^2}{4}\pi^2t\} \iint_{\indom} &u_0(x,y)\sin(\frac{k_1\pi x}{2})\cos(\frac{k_2\pi y}{2})\,\mathrm{d}x\mathrm{d}y \\&+ \int_0^t\exp\{-\mu\frac{k_1^2+k_2^2}{4}\pi^2(t-s)\}f_{k_1,k_2}^{(2)}(s)\,\mathrm{d}s \\
C_{k_1,k_2}(t) = \exp\{-\mu\frac{k_1^2 + k_2^2}{4}\pi^2t\} \iint_{\indom} &u_0(x,y)\cos(\frac{k_1\pi x}{2})\sin(\frac{k_2\pi y}{2})\,\mathrm{d}x\mathrm{d}y \\&+ \int_0^t\exp\{-\mu\frac{k_1^2+k_2^2}{4}\pi^2(t-s)\}f_{k_1,k_2}^{(3)}(s)\,\mathrm{d}s \\
D_{k_1,k_2}(t) = \exp\{-\mu\frac{k_1^2 + k_2^2}{4}\pi^2t\} \iint_{\indom} &u_0(x,y)\cos(\frac{k_1\pi x}{2})\cos(\frac{k_2\pi y}{2})\,\mathrm{d}x\mathrm{d}y \\&+ \int_0^t\exp\{-\mu\frac{k_1^2+k_2^2}{4}\pi^2(t-s)\}f_{k_1,k_2}^{(4)}(s)\,\mathrm{d}s.
\end{align*}
From boundary conditions, $A_{k_1,k_2}$ is for $k_1$ even and $k_2$ even, $B_{k_1,k_2}$ is for $k_1$ even and $k_2$ odd, $C_{k_1,k_2}$ is for $k_1$ odd and $k_2$ even, $D_{k_1,k_2}$ is for $k_1$ odd and $k_2$ odd. 

\subsubsection{Advection-diffusion Equation with External Source}
The two-dimensional advection-diffusion equation with external source is
\begin{equation}\label{lidar_eq:ad_diff}
\frac{\partial u}{\partial t} + c_1\frac{\partial u}{\partial x} + c_2\frac{\partial u}{\partial y} - \mu\big(\frac{\partial^2 u}{\partial x^2} + \frac{\partial^2 u}{\partial y^2}\big) = f(x,y,t),\quad -1<x, y<1, \quad t\in[0,T].
\end{equation}
where $c = (c_1, c_2)$ is the velocity constant and $\mu$ is the diffusivity. A change of variables $u(x,y,t) = v(x,y,t)e^{\alpha x +\beta  y +\gamma t}$ transforms \eqref{lidar_eq:ad_diff} into a diffusion equation that we already solved, and here are the details: We rewrite the equation \eqref{lidar_eq:ad_diff} in terms of $v(x,y,t)$:
\begin{align}
\frac{\partial v}{\partial t} + (c_1-2\mu\alpha)\frac{\partial v}{\partial x} + (c_2&-2\mu\beta)\frac{\partial v}{\partial y} + (\gamma + c_1\alpha + c_2\beta - \mu\alpha^2 - \mu\beta^2)v - \mu\big(\frac{\partial^2 v}{\partial x^2} + \frac{\partial^2 v}{\partial x^2}\big) \\
=&\ f(x,y,t)\exp\{-\alpha x - \beta y - \gamma t\}.
\end{align}
Set some coefficients to zero: 
$$
\begin{cases}
c_1 - 2\mu\alpha = 0\\
c_2 - 2\mu\beta = 0\\
\gamma + c_1\alpha + c_2\beta - \mu\alpha^2 - \mu\beta^2 = 0
\end{cases}
\quad\Rightarrow\quad
\begin{cases}
\alpha = c_1/2\mu\\
\beta = c_2/2\mu\\
\gamma = -(c_1^2 + c_2^2)/4\mu.
\end{cases}
$$
Now $v(x, y, t)$ satisfies the heat equation with homogeneous Dirichlet conditions: 
\begin{equation}
\begin{cases}
v_t - \mu(v_{xx} + v_{yy}) = \tilde{f}(x,y,t)\\
v(x,y,t) = 0, \mbox{ for }(x, y)\mbox{ on the boundary.}\\
v_0(x, y) =  \exp\{- c_1x/2\mu - c_2y/2\mu\}u_0(x,y),
\end{cases}
\end{equation}
where $\tilde{f}(x,y,t) = \exp\{(c_1^2 + c_2^2)t/4\mu - c_1x/2\mu - c_2y/2\mu\}f(x,y,t) $. Its relation to $u(x,y,t)$ is given by
\begin{equation}
u(x,y,t) = \exp\{-(c_1^2 + c_2^2)t/4\mu + c_1x/2\mu + c_2y/2\mu\}v(x,y,t).
\end{equation}
Based on the result \eqref{lidar_eq:sol_heat_force} on heat equation, the solution $u(x,y,t)$ is:
\begin{align}\label{lidar_eq:sol_to_ad}
u(x,y,t) = &\exp\{-(c_1^2 + c_2^2)t/4\mu + c_1x/2\mu + c_2y/2\mu\} \sum_{k_1\ge0}\sum_{k_2\ge0} \\
&\Big\{ A_{k_1,k_2}(t)\sin(\frac{k_1\pi x}{2})\sin(\frac{k_2\pi y}{2}) + B_{k_1,k_2}(t)\sin(\frac{k_1\pi x}{2})\cos(\frac{k_2\pi y}{2}) \\
&+ C_{k_1,k_2}(t)\cos(\frac{k_1\pi x}{2})\sin(\frac{k_2\pi y}{2}) + D_{k_1,k_2}(t)\cos(\frac{k_1\pi x}{2})\cos(\frac{k_2\pi y}{2}) \Big\}.
\end{align}
where
\begin{align*}
A_{k_1,k_2}(t) = \phi_{k_1,k_2}^{(1)}+ \int_0^t \exp\{-\mu\frac{k_1^2 + k_2^2}{2}\pi^2(t-s)\}\tilde{f}_{k_1,k_2}^{(1)}(s)\,\mathrm{d}s, \quad\mbox{for }k_1\mbox{ even, }k_2\mbox{ even;}\\
B_{k_1,k_2}(t) = \phi_{k_1,k_2}^{(2)}+ \int_0^t \exp\{-\mu\frac{k_1^2 + k_2^2}{2}\pi^2(t-s)\}\tilde{f}_{k_1,k_2}^{(2)}(s)\,\mathrm{d}s,\quad\mbox{for }k_1\mbox{ even, }k_2\mbox{ odd;} \\
C_{k_1,k_2}(t) = \phi_{k_1,k_2}^{(3)}+ \int_0^t \exp\{-\mu\frac{k_1^2 + k_2^2}{2}\pi^2(t-s)\}\tilde{f}_{k_1,k_2}^{(3)}(s)\,\mathrm{d}s,\quad\mbox{for }k_1\mbox{ odd, }k_2\mbox{ even;}\\
D_{k_1,k_2}(t) = \phi_{k_1,k_2}^{(4)}+ \int_0^t \exp\{-\mu\frac{k_1^2 + k_2^2}{2}\pi^2(t-s)\}\tilde{f}_{k_1,k_2}^{(4)}(s)\,\mathrm{d}s,
\quad\mbox{for }k_1\mbox{ odd, }k_2\mbox{ odd;}
\end{align*}

$\{\phi_{k_1,k_2}^{(i)}\}$ is related to the Fourier coefficients of the initial condition $v_0$ (or equivalently, $u_0$) as in \eqref{lidar_eq:sol_heat_force}. We see that the solution $u(x,y,t)$ is an additive sum of two components: one from the initial condition $u_0(x,y)$, and the other from the external source $f(x,y,t)$. The two sources act independently on the solution, so the external source does not play a role in the selection of sensor locations, when we use a Bayesian framework of Gaussian distributions to infer the initial condition from time-space measurements.

\subsection{Numerical Results}
We provide numerical results on selecting the optimal sensing directions to estimate the initial condition of the advection-diffusion equation. Here is the problem description: suppose a lidar is sitting at the origin of a unit circle ($\outdom$), and it collects data $u(x,y,t)$ by sending out laser beams and detects the reflections; we need to determine the optimal directions for the lidar to release the beams and collect data. Our parameter-to-observable mapping is directly from the solution to \eqref{lidar_eq:ad_diff}, which is an integral equation
\begin{equation}
u(x,y,t) = \mathcal{F}(u_0) = \iint_{[-1,1]\times[-1,1]} \mathrm{F}(x,y,t)u_0(x,y)\,dxdy
\end{equation}
where $\mathrm{F}$ is given in the solution \eqref{lidar_eq:sol_to_ad}. For discretizations, we divide the angle of $2\pi$ into $n_d$ parts so that the circle has $n_d$ sectors with the same area, and each beam goes across the center of each sector. We also discretize the radius into $n_r$ parts with equal length. A weight variable is attached to each sector, and discretization points along the same radius have the same weight. The goal is to select a proportion of sectors and measure data on those radiuses, in order to best infer $u_0$ which is defined on a slightly larger square domain $[-1,1]\times[-1,1]$ ($\indom$), which is discretized by regular grid of size $n_x\times n_x$. 

\begin{figure}[H]
\centering
\begin{minipage}{.33\textwidth}
  \centering
  \includegraphics[height=1.4in, width=2.0in]{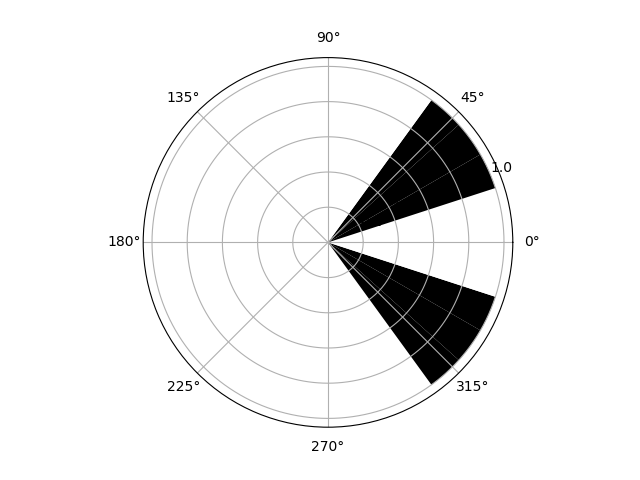}
\end{minipage}%
\begin{minipage}{.33\textwidth}
  \centering
  \includegraphics[height=1.4in, width=2.0in]{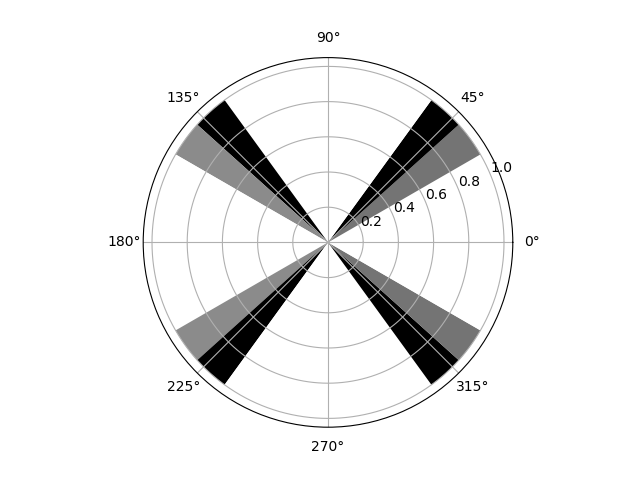}
\end{minipage}
\begin{minipage}{.33\textwidth}
  \centering
  \includegraphics[height=1.4in, width=2.0in]{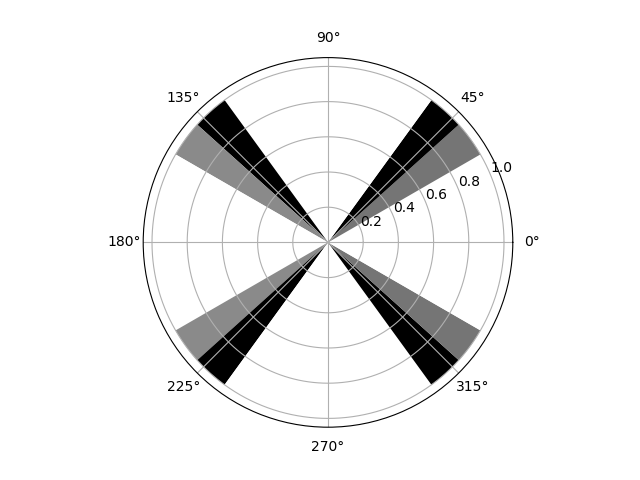}
\end{minipage}
\caption{Relaxed solution (with full $F$) for $p = 1, 2, 3$ respectively}\label{fig:relax}
\end{figure}

\begin{figure}[H]
\centering
\begin{minipage}{.33\textwidth}
  \centering
  \includegraphics[height=1.4in, width=2.0in]{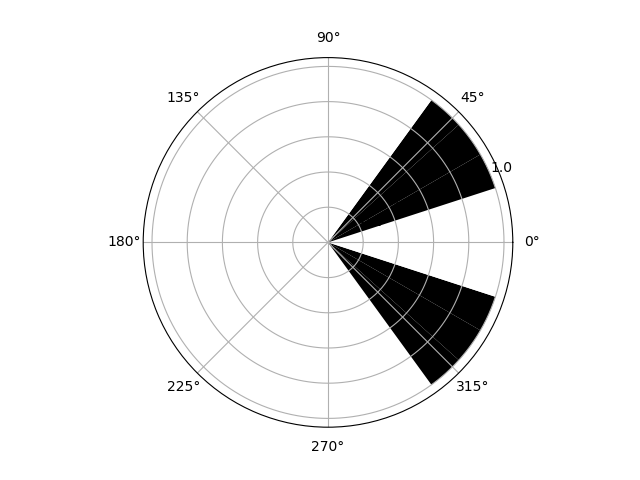}
\end{minipage}%
\begin{minipage}{.33\textwidth}
  \centering
  \includegraphics[height=1.4in, width=2.0in]{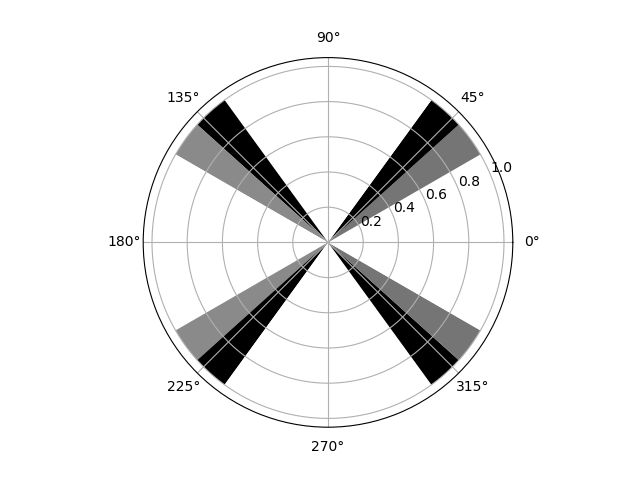}
\end{minipage}
\begin{minipage}{.33\textwidth}
  \centering
  \includegraphics[height=1.4in, width=2.0in]{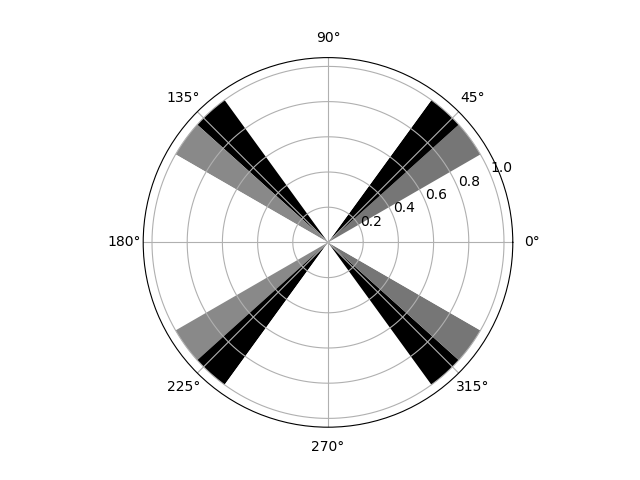}
\end{minipage}
\caption{SQP solution (with low-rank $F$) for $p = 1, 2, 3$ respectively}
\end{figure}

\begin{figure}[H]
\centering
\begin{minipage}{.33\textwidth}
  \centering
  \includegraphics[height=1.4in, width=2.0in]{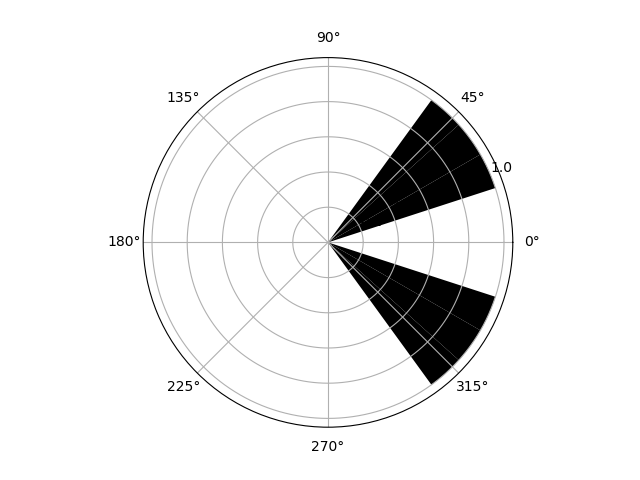}
\end{minipage}%
\begin{minipage}{.33\textwidth}
  \centering
  \includegraphics[height=1.4in, width=2.0in]{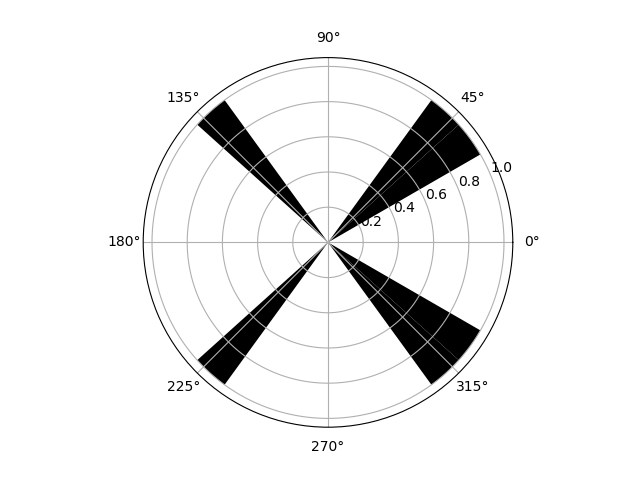}
\end{minipage}
\begin{minipage}{.33\textwidth}
  \centering
  \includegraphics[height=1.4in, width=2.0in]{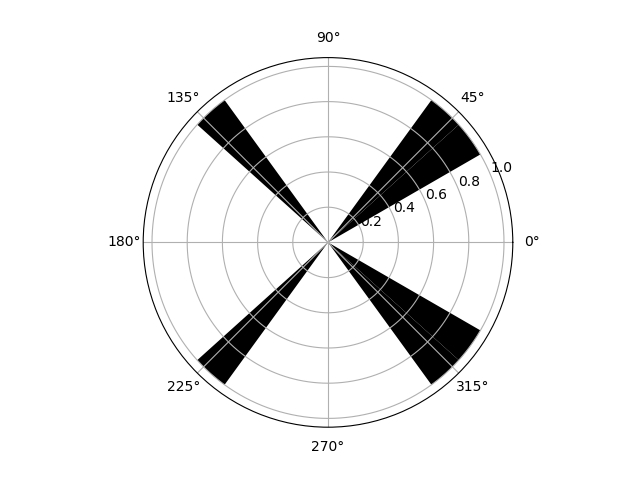}
\end{minipage}
\caption{Sum-up Rounding for $p = 1, 2, 3$ respectively (based on relaxed solutions)}
\end{figure}

The constants we choose in the equation \eqref{lidar_eq:ad_diff} are $c_1 = 0.1, c_2 = 0, \mu = 1.0, T = 1, n_t = 5, r = 0.2$, and the noise ratio $\alpha=\rsigma$ is 0.01. We remind the meaning of these constants: $c_1, c_2$ are the velocities along the $x$ and $y$ axis respectively, $\mu$ is the diffusivity constant, $n_t$ is a fixed integer denoting the number of measurements in $[0, T]$, and $r$ is the proportion of selected sectors. The covariance matrix in time is set to be identity at the moment. For the results below, $n_d=n_r=n_x=30$, and the velocity $(c_1,c_2)$ is going from left to right (the advection term can be thought of as air movement or wind when $u$ is the concentration of a substance in the air). the problem and the design are symmetric to the x axis (see Figure. \ref{fig:relax}). The $\epsilon$ in the stopping criterion of Algorithm \ref{sqp_algo:sqp} is $10^{-3}$.

Because the solution to advection-diffusion equation in \eqref{lidar_eq:sol_to_ad} is an integral of an infinite sum, we truncate the sum by specifying a hyperparameter $p$ and take the dominating terms with $k_1, k_2\le p$. We can determine the value of $p$ from a sanity check (see Figure. \ref{fig:sanity}), where we try to recover the initial state by looking at the truncated solution at $t=0$. We observe that when $p=3$, the values do not change further, and we use $p = 3$. Note this choice of $p$ is subject to the choice of $c_1, c_2$ and $\mu$ in the equation, especially when $\mu$ is small, a larger value of $p$ is required.

\begin{figure}[!ht]
   \centering
   \subfloat[][]{\includegraphics[width=.33\textwidth]{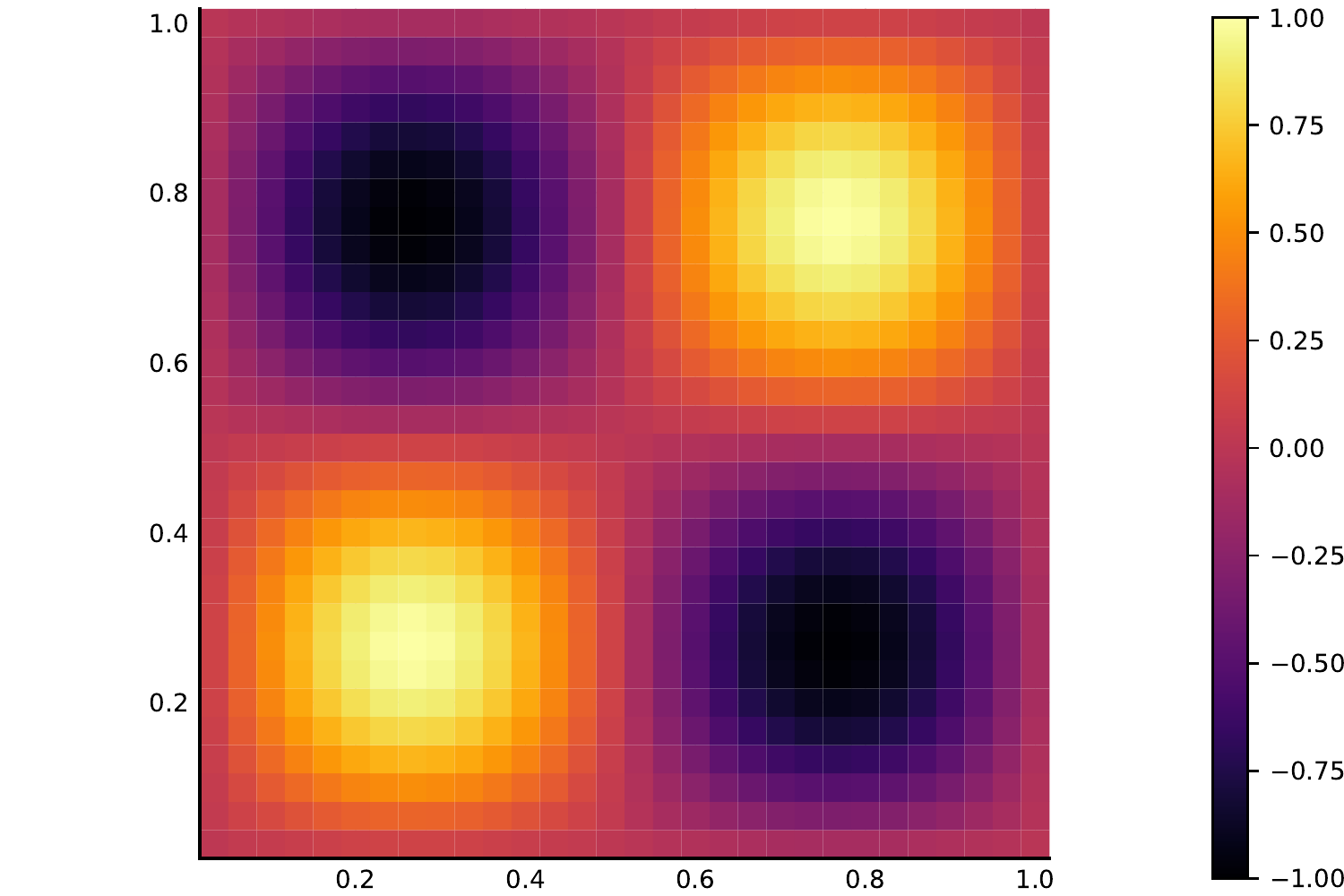}}\quad
   \subfloat[][]{\includegraphics[width=.33\textwidth]{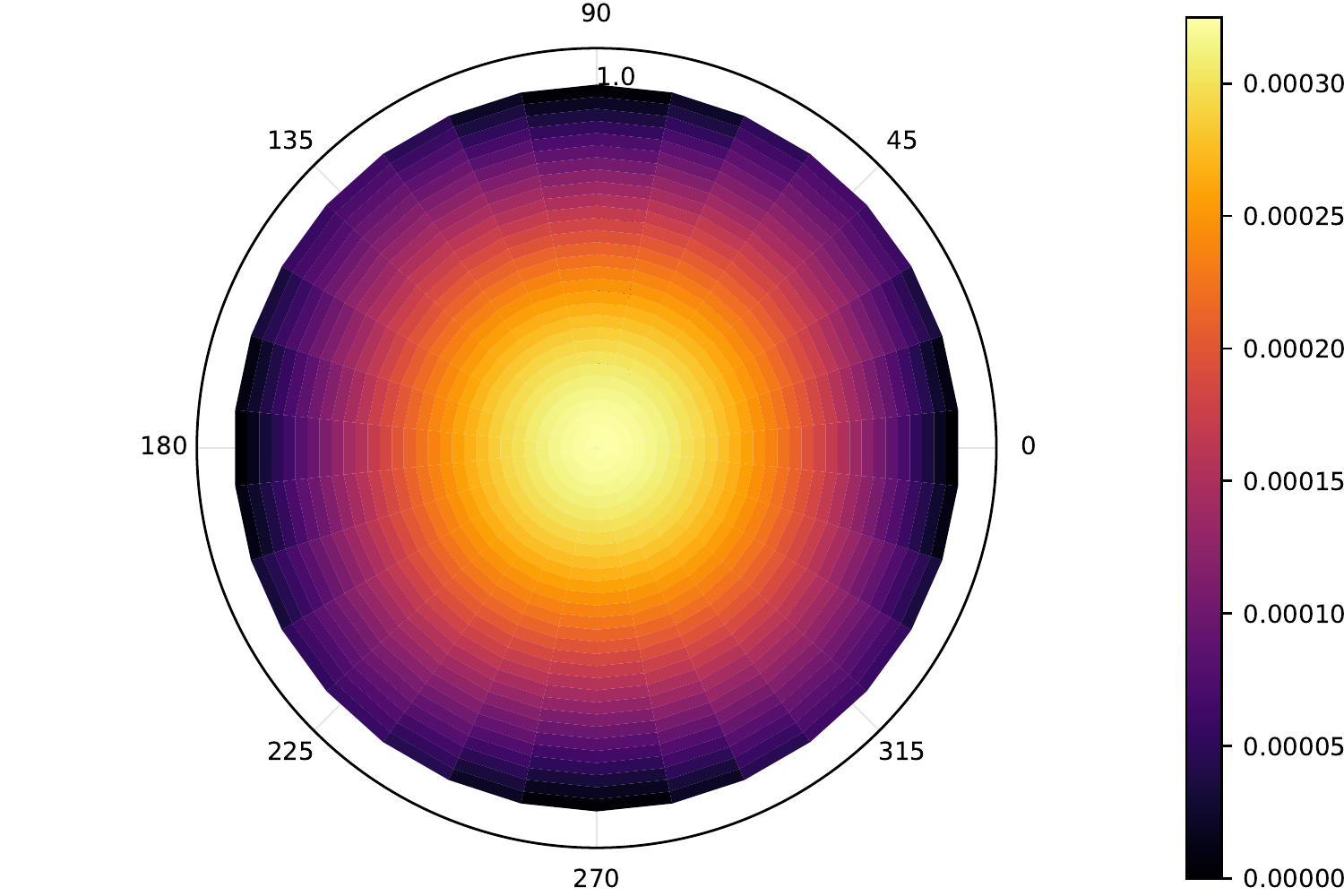}}\\
   \subfloat[][]{\includegraphics[width=.33\textwidth]{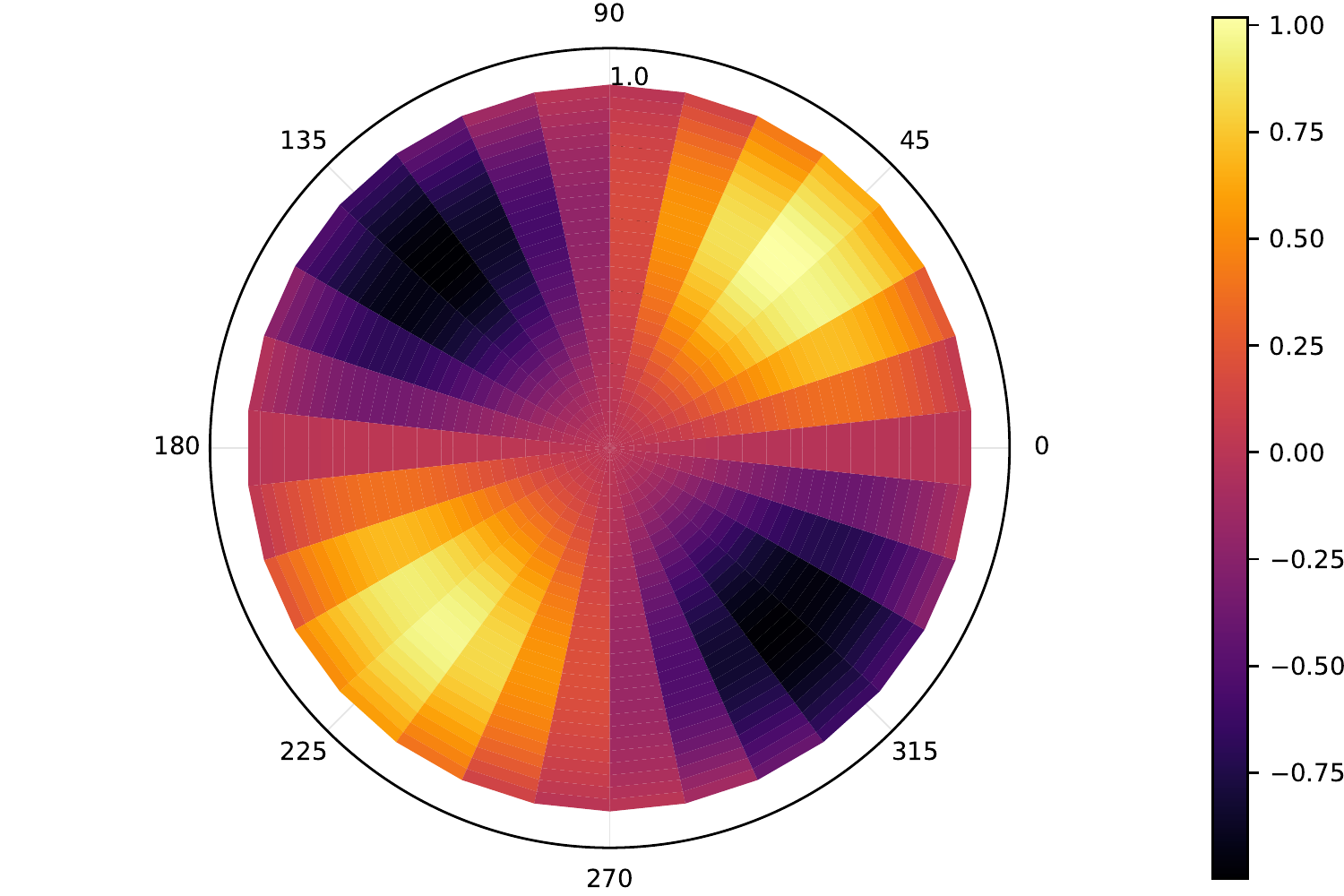}}\quad
   \subfloat[][]{\includegraphics[width=.33\textwidth]{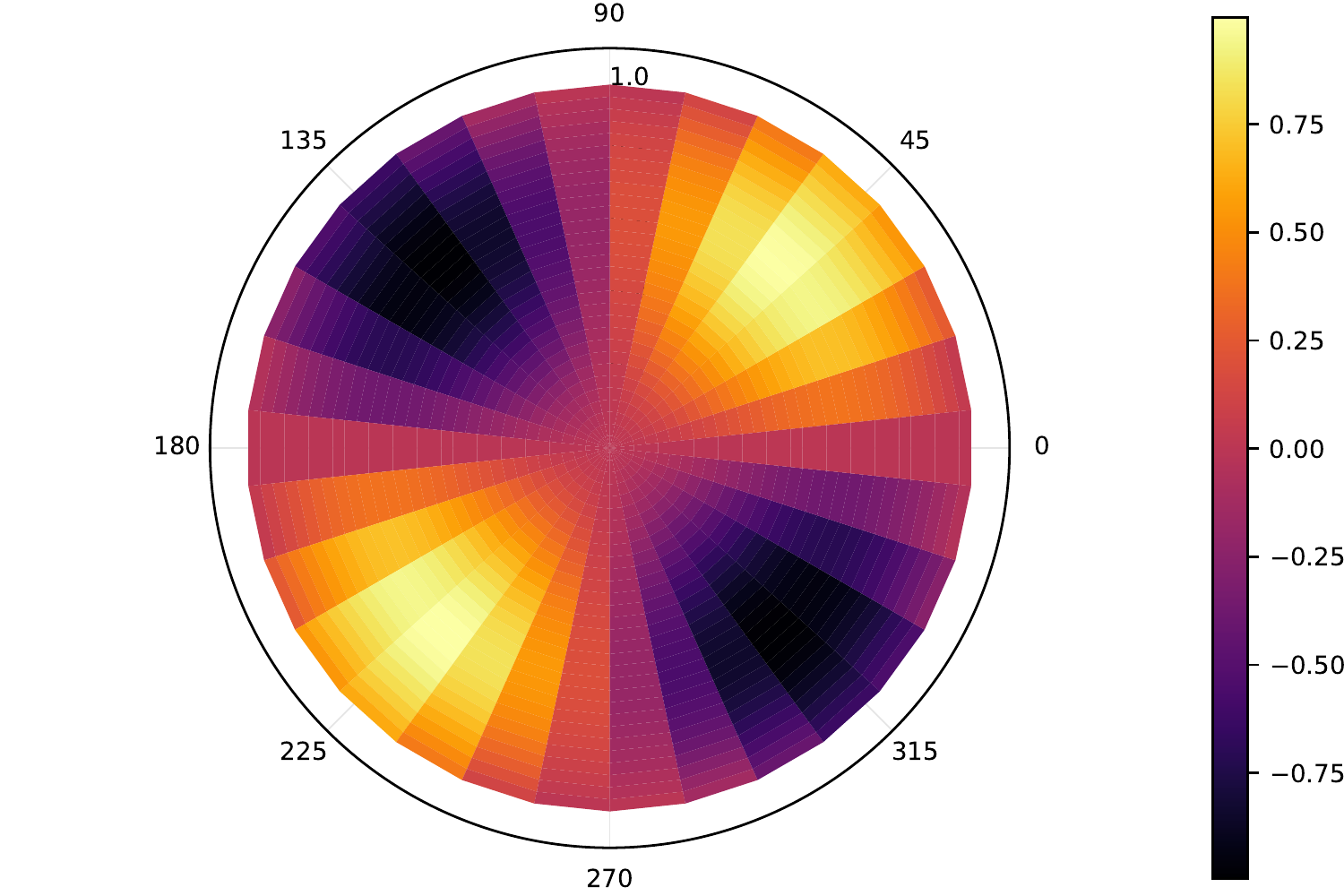}}
   \caption{(a) Initial state: $u_0(x,y) = \sin(\pi x)(\pi y)$ on $[-1,1]\times[-1,1]$; (b) Recover the initial state using dominant terms with $k_1,k_2\le 1$; (c) Recover the initial state using dominant terms with $k_1,k_2\le 2$; (d) Recover the initial state using dominant terms with $k_1,k_2\le 3$.}
   \label{fig:sanity}
\end{figure}

We examine the performance of SQP by looking at the computation time in comparison with the $\textit{Ipopt}$ package in Julia, and its integrality gap in the objective, i.e. the difference in \eqref{error_eq:gap_full} with the full $F$ and the approximated $F_s$ respectively. We use different numbers of interpolation points $c\cdot\log(n)$ by choosing the constant $c = 1, 2, 4, 8$, and let $n_d = n_r=n_x$. When $n_r = 30$, it takes $\textit{Ipopt}$ about 1.5 hours to compute the solution (see Figure \ref{fig:comp_time}), while SQP needs less than a minute to get a sufficiently good approximation. We would like to mention that the stopping criterion in the $\textit{Ipopt}$ package is $10^{-6}$ in order to get the true minimum, and it is ``unfair'' to compare the computation time directly with the SQP algorithm ($\epsilon=10^{-3}$). However, as the stopping criterion only affects the number of iterations, the computation time of the exact method is still proportional to what we have seen in Figure \ref{fig:comp_time}, and it is much slower than SQP. 

In Figure \ref{fig:gap_lowrank}, the result for $c=1$ is missing because the gap is identically zero which implies the relaxed problem \eqref{eq:doe_lowrank} has an integer solution, and SUR found it. However, this does not suggest we should choose $c=1$, because there are few interpolation points and $F_s$ is not a good approximation of $F$, so the design is not necessarily good (see Figure \ref{fig:gap_full}), and actually the design it returns is uniform everywhere. When we increase $c$ or the number of interpolation points, the integrality gap becomes smaller. In \S\ref{sec:error}, we show the integrality gap in the objective with the full $F$ converges to zero, which is illustrated in Figure \ref{fig:gap_full}. We observe the gap decreases faster for larger $c$, but it is not monotone.

\begin{figure}[H]
\centering
\begin{minipage}[t]{.33\textwidth}
  \centering
  \includegraphics[height=2.0in, width=2.0in]{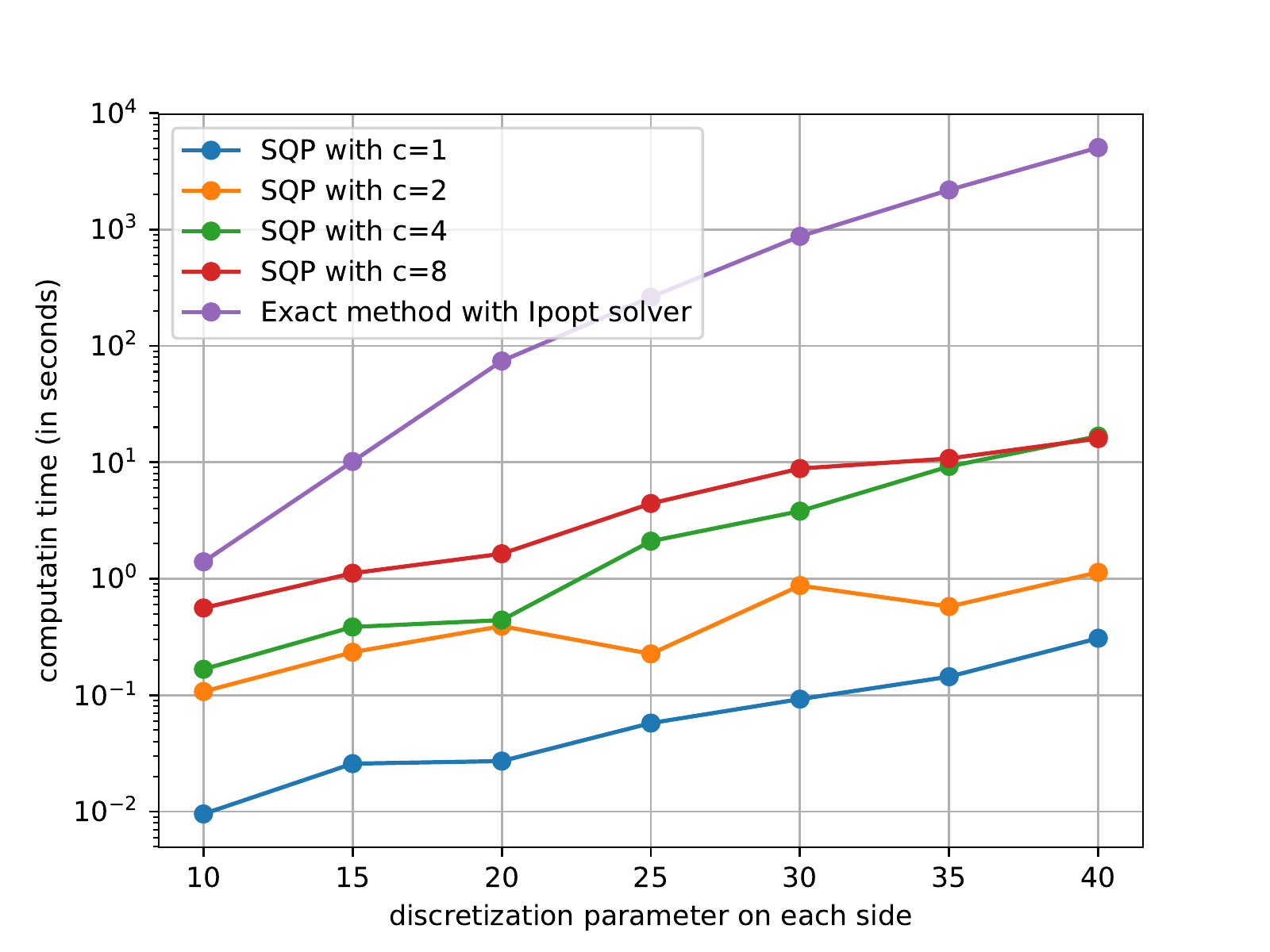}
  \captionsetup{justification=centering}
  \captionof{figure}{Computation time}
  \label{fig:comp_time}
\end{minipage}%
\begin{minipage}[t]{.33\textwidth}
  \centering
  \includegraphics[height=2.0in, width=2.0in]{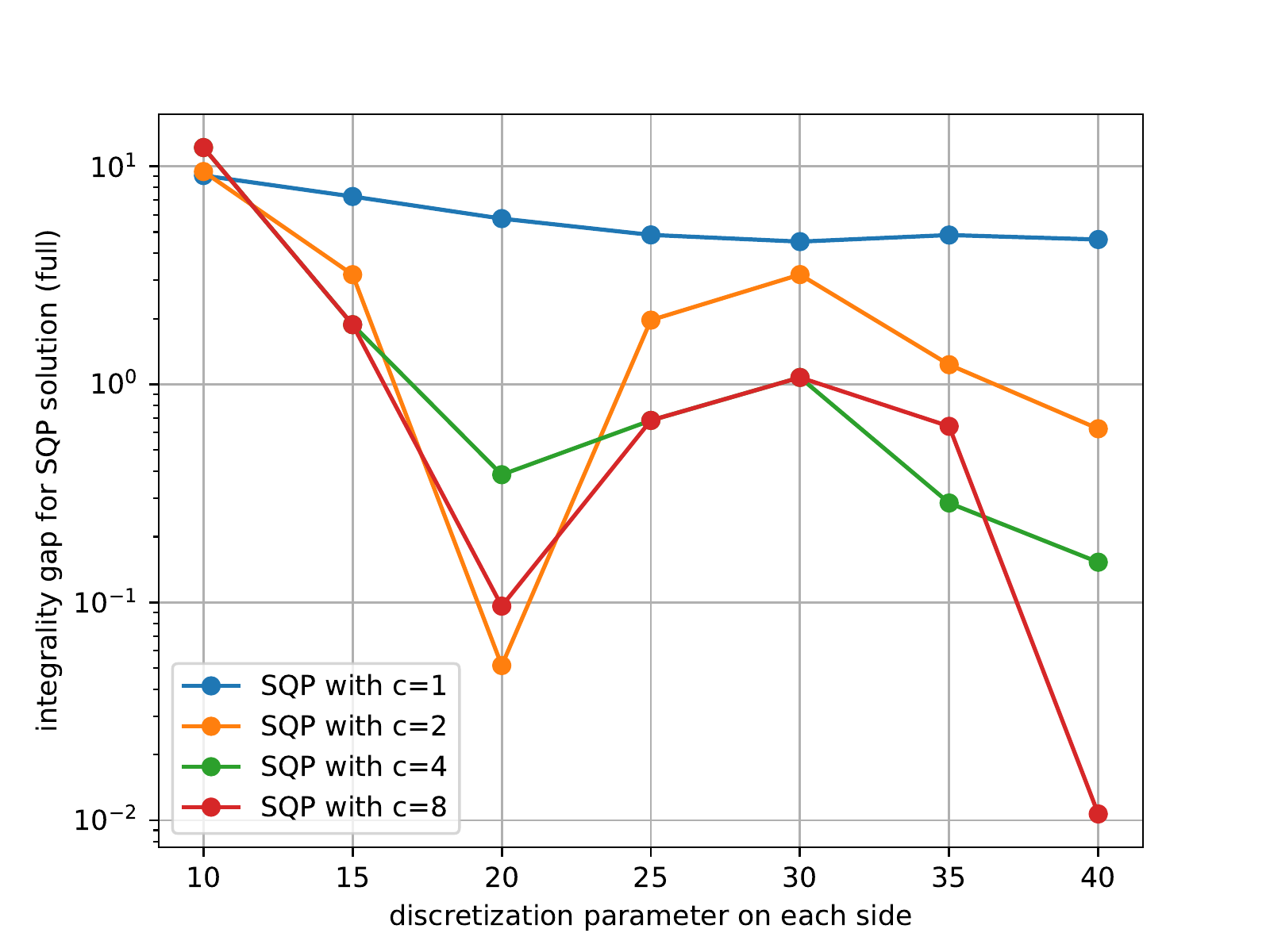}
  \captionsetup{justification=centering}
  \captionof{figure}{Integrality gap\\ (full $F$)}
  \label{fig:gap_full}
\end{minipage}
\begin{minipage}[t]{.33\textwidth}
  \centering
  \includegraphics[height=2.0in, width=2.0in]{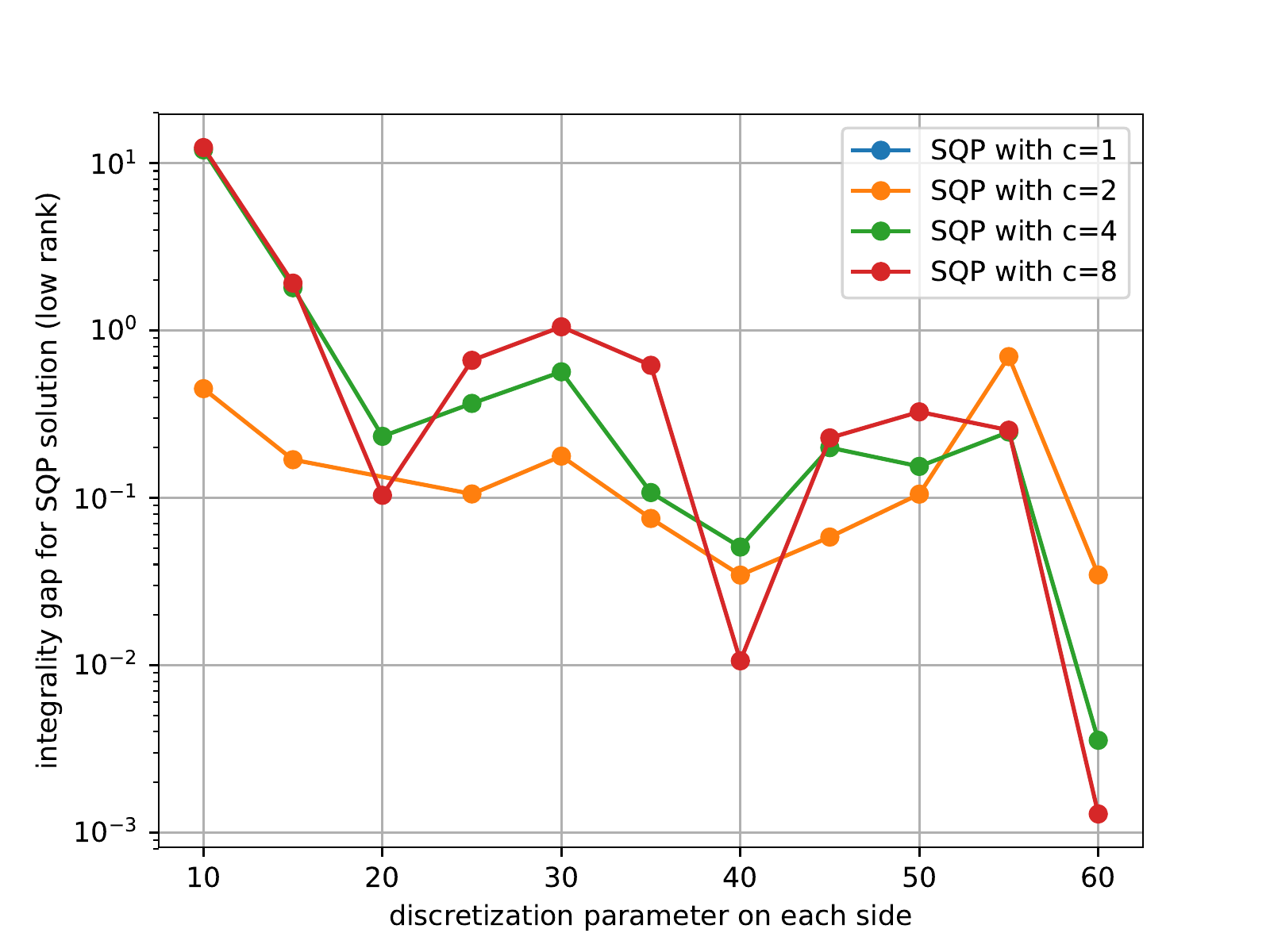}
  \captionsetup{justification=centering}
  \captionof{figure}{integrality gap\\ (low rank $F$)}
  \label{fig:gap_lowrank}
\end{minipage}
\end{figure}

To see the effect of $c_1, c_2$ and $\mu$ on the optimal sensing directions, we conduct more experiments with the $\textit{Ipopt}$ package in Julia. The following figures give the exact relaxed solution for varying values of $c_1$ and $\mu$, but fixed $c_2=0$ and $p=3$. Again, $n_d = n_r = 30$.

\begin{figure}[H]
\centering
\begin{minipage}{.33\textwidth}
  \centering
  \includegraphics[height=1.5in, width=2.0in]{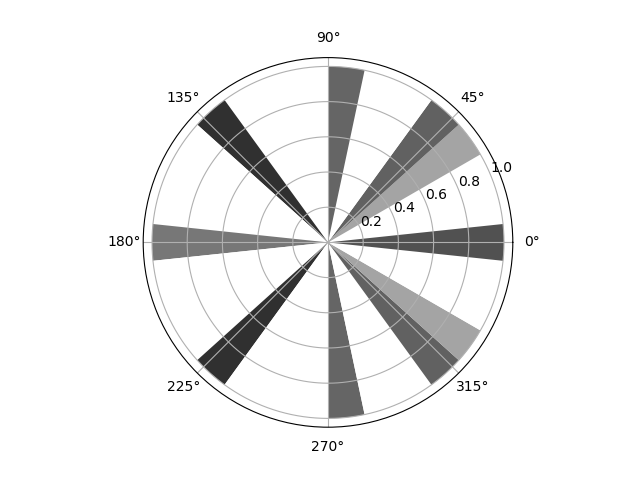}
\end{minipage}%
\begin{minipage}{.33\textwidth}
  \centering
  \includegraphics[height=1.5in, width=2.0in]{graphs/p_3_rel.png}
\end{minipage}
\begin{minipage}{.33\textwidth}
  \centering
  \includegraphics[height=1.5in, width=2.0in]{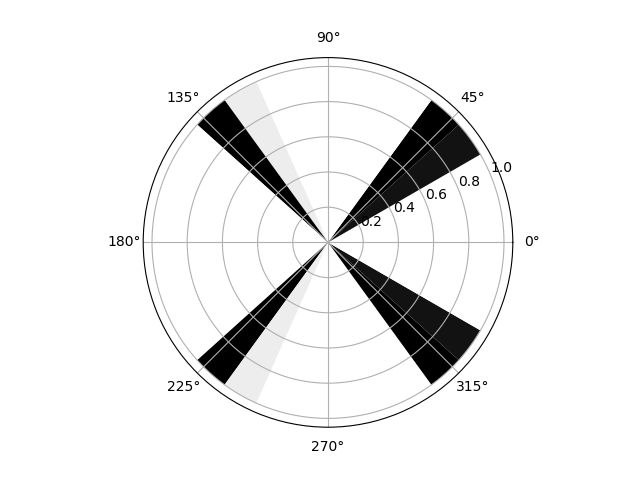}
\end{minipage}
\caption{Dependence of sensing direction on $c_1$ and $\mu$ when wind blows $\to$. From left to right: (1) $c_1 = 0.1, \mu = 0.1$; (2) $c_1 = 0.1, \mu = 1.0$; (3) $c_1 = 1.0, \mu = 1.0$. }\label{algo_fig:exact_c}
\end{figure}

\begin{figure}[H]
\centering
\begin{minipage}{.22\textwidth}
  \centering
  \includegraphics[height=1.2in, width=1.5in]{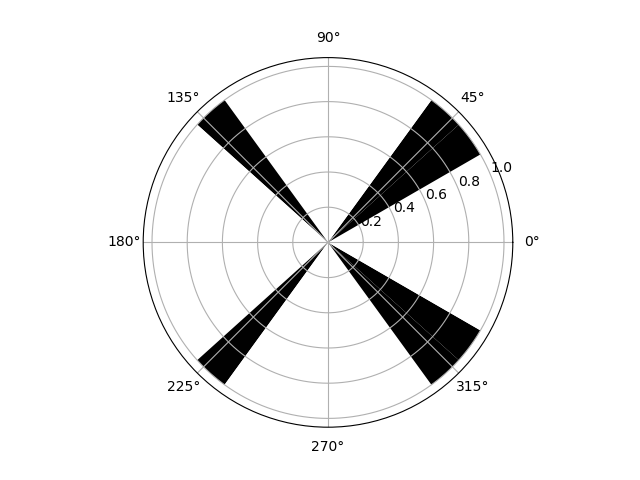}
\end{minipage}%
\begin{minipage}{.22\textwidth}
  \centering
  \includegraphics[height=1.2in, width=1.5in]{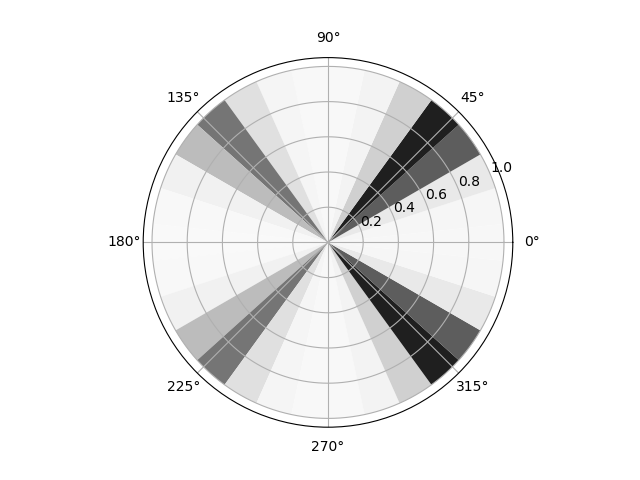}
\end{minipage}
\begin{minipage}{.22\textwidth}
  \centering
  \includegraphics[height=1.2in, width=1.5in]{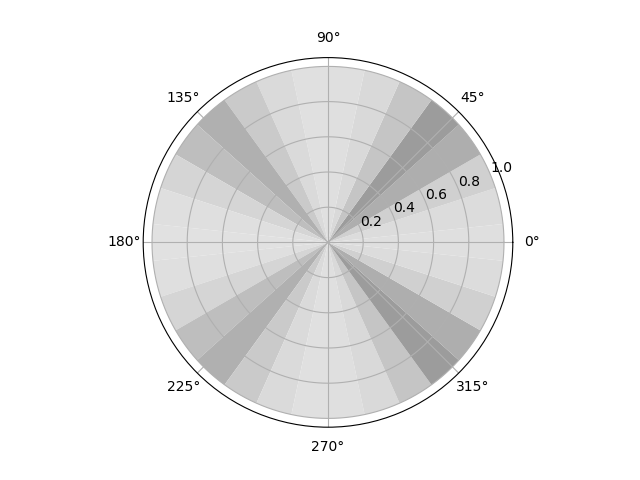}
\end{minipage}
\begin{minipage}{.22\textwidth}
  \centering
  \includegraphics[height=1.2in, width=1.5in]{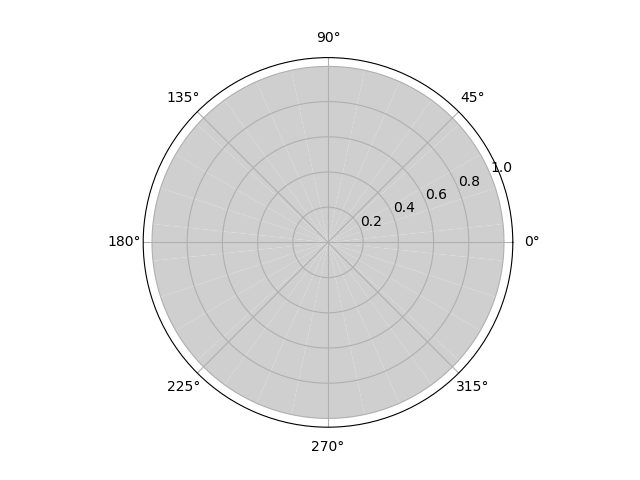}
\end{minipage}
\caption{Dependence of sensing direction on $\mu$ when $c_1=0.1$ and wind blows $\to$. From left to right: $\mu = 5.0, 7.0, 8.0, 10.0$. }\label{algo_fig:exact_mu}
\end{figure}

Based on Figure. \ref{algo_fig:exact_c} and Figure. \ref{algo_fig:exact_mu}, we find that
\begin{itemize}
\item when the wind (with velocity $(c_1, c_2)$) moves faster, more sensing directions are chosen towards the wind;
\item when it is less diffusive (small values of $\mu$), the sensing directions spread out more;
\item when the diffusivity $\mu$ is large, the relaxed sensing weights gets blurred.
\end{itemize}

With SQP, we are able to run problems of larger sizes ($n_d = n_r = 80$) and we change the wind direction form $\to$ to $ \nearrow$. The relaxed sensing directions are given below, which confirms that more sensing directions should be selected towards the wind when it moves faster.

\begin{figure}[H]
\centering
\begin{minipage}{.33\textwidth}
  \centering
  \includegraphics[height=1.7in, width=2.2in]{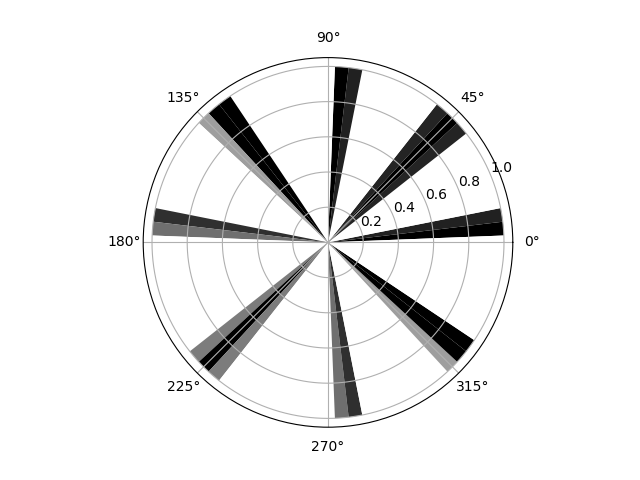}
\end{minipage}%
\begin{minipage}{.33\textwidth}
  \centering
  \includegraphics[height=1.7in, width=2.2in]{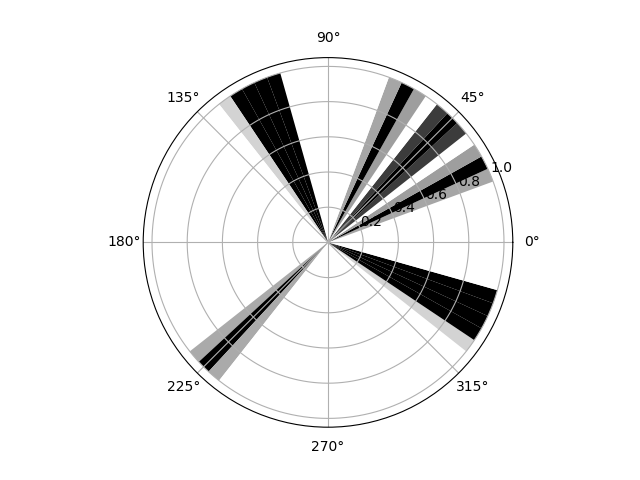}
\end{minipage}
\begin{minipage}{.33\textwidth}
  \centering
  \includegraphics[height=1.7in, width=2.2in]{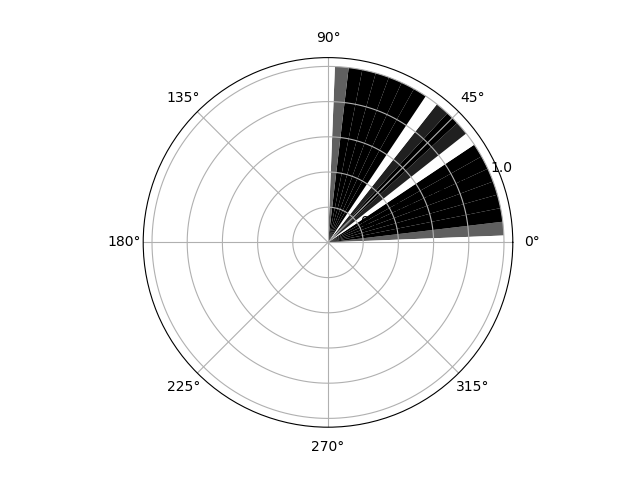}
\end{minipage}
\caption{Sensing direction for increasing wind speed with $\mu=0.1$, wind direction$ \nearrow$. From left to right: (1) $c_1 = c_2 = 0.1$; (2) $c_1 = c_2 = 0.5$; (3) $c_1 = c_2 = 1.0$.}\label{algo_fig:sqp_c}
\end{figure}

From the solution to the advection-diffusion equation, we know that for a larger value of $\mu$, $p$ imposes less effect on the sensing directions. But when $\mu$ is small, such as 0.1, the design is likely to depend on $p$, and adding $p$ makes the design more ``diffusive'', and the selected directions covers a wider range of angles, see Figure \ref{algo_fig:sqp_p}.

\begin{figure}[H]
\centering
\begin{minipage}{.33\textwidth}
  \centering
  \includegraphics[height=1.7in, width=2.2in]{graphs/c_1mu_e-1_16.png}
\end{minipage}
\begin{minipage}{.33\textwidth}
  \centering
  \includegraphics[height=1.7in, width=2.2in]{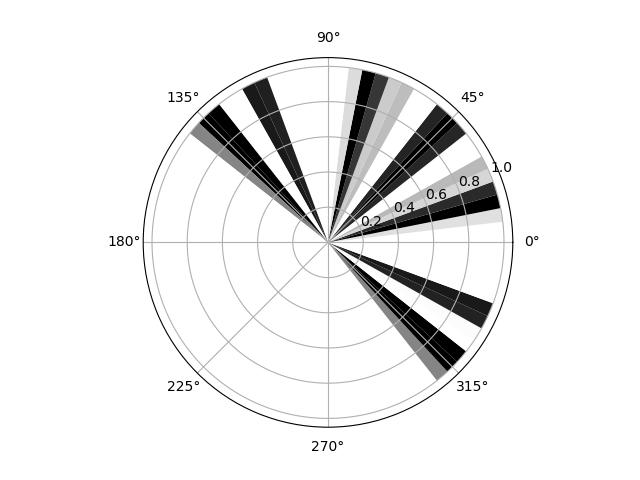}
\end{minipage}%
\begin{minipage}{.33\textwidth}
  \centering
  \includegraphics[height=1.7in, width=2.2in]{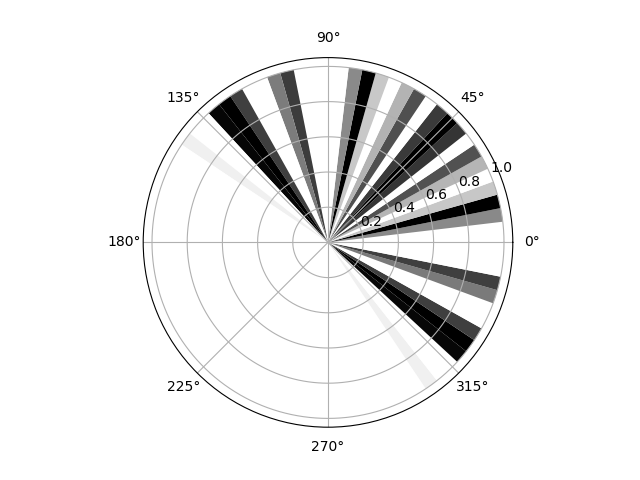}
\end{minipage}
\captionof{figure}{Sensing direction for increasing $p$ with $c_1=c_2=1.0$ and $\mu = 0.1$, wind direction$ \nearrow$. From left to right: $p = 3, 5, 10$.}\label{algo_fig:sqp_p}
\end{figure}

\section{Discussion}\label{sec:disc}

We present an interpolation-based SQP algorithm to solve a convex optimization problem stemming from optimal sensor placement. Most of the algorithm is implemented by hand in Julia, except the singular value decomposition and eigenvalue decomposition. 

The algorithm returns a reasonably good solution after the first several iterations. In the Hessian approximation, $H_s$ is positive semi-definite, which ensures the step returned by each iteration is a descent direction of the objective. However, the decay in the objective can be slow after several initial iterations due to the low-rank structure of $H_s$. We encounter this in the LIDAR problem, if we use a more stringent tolerance $\epsilon$ in the stopping criterion, the iteration number can increase significantly, although the gain in the objective is very limited. This user-defined $\epsilon$ depends on the particular setting of the problems we are trying to solve.

One way to improve the current algorithm is to incorporate high performance computing, since the algorithm involves lots of matrix-vector and vector-vector multiplications. For example, we construct the $F$ matrix, (large size, low-rank, dense) on different processors and run the linear algebra in parallel. The parallel algorithm will consider partitioning the matrix and related computations and how to share information across the processors (see \citep{gallivan1990parallel}), which is the main direction for future work. 

\bibliographystyle{gOMS}
\bibliography{ref}

\begin{thebibliography}{10}
\newcommand{\noopsort}[1]{}
\newcommand{\printfirst}[2]{#1}
\newcommand{\singleletter}[1]{#1}
\newcommand{\switchargs}[2]{#2#1}
\providecommand{\url}[1]{\normalfont{#1}}
\providecommand{\urlprefix}{Available at }

\bibitem{anderian2018efficient}
A. Alexanderian and A.K. Saibaba, \emph{Efficient d-optimal design of
  experiments for infinite-dimensional bayesian linear inverse problems}, SIAM
  Journal on Scientific Computing 40 (2018), pp. A2956--A2985.

\bibitem{alexanderian2014optimal}
A. Alexanderian, N. Petra, G. Stadler, and O. Ghattas, \emph{A-optimal design
  of experiments for infinite-dimensional {B}ayesian linear inverse problems
  with regularized $l_0$-sparsification}, SIAM Journal on Scientific Computing
  36 (2014), pp. A2122--A2148.

\bibitem{sandia2}
J. Berry, W.E. Hart, C.A. Phillips, J.G. Uber, and J.P. Watson, \emph{Sensor
  placement in municipal water networks with temporal integer programming
  models}, Journal of Water Resources Planning and Management 132 (2006), pp.
  218--224.

\bibitem{boyd}
S. Boyd and L. Vandenberghe, \emph{Convex optimization}, Cambridge University
  Press, 2004.

\bibitem{LebCons}
L. Brutman, \emph{Lebesgue functions for polynomial interpolation - a survey},
  Annals of Numerical Mathematics, 1996.

\bibitem{gallivan1990parallel}
K.A. Gallivan, R.J. Plemmons, and A.H. Sameh, \emph{Parallel algorithms for
  dense linear algebra computations}, SIAM review 32 (1990), pp. 54--135.

\bibitem{greenbaum2012numerical}
A. Greenbaum and T.P. Chartier, \emph{Numerical methods: design, analysis, and
  computer implementation of algorithms}, Princeton University Press, 2012.

\bibitem{krausewater}
A. Krause, J. Leskovec, C. Guestrin, J. VanBriesen, and C. Faloutsos,
  \emph{Efficient sensor placement optimization for securing large water
  distribution networks}, Journal of Water Resources Planning and Management
  134 (2008), pp. 516--526.

\bibitem{Nocedal_book}
J. Nocedal and S.J. Wright, \emph{Numerical Optimization}, 2nd ed., Springer,
  New York, 2006.

\bibitem{PetraMartinStadlerEtAl14}
N. Petra, J. Martin, G. Stadler, and O. Ghattas, \emph{A computational
  framework for infinite-dimensional {B}ayesian inverse problems: {P}art {II}.
  {S}tochastic {N}ewton {MCMC} with application to ice sheet inverse problems},
  SIAM Journal on Scientific Computing 36 (2014), pp. A1525--A1555.

\bibitem{puke}
F. Pukelsheim, \emph{Optimal design of experiments}, Classics in Applied
  Mathematics, SIAM, 2006.

\bibitem{spantini2015optimal}
A. Spantini, A. Solonen, T. Cui, J. Martin, L. Tenorio, and Y. Marzouk,
  \emph{Optimal low-rank approximations of bayesian linear inverse problems},
  SIAM Journal on Scientific Computing 37 (2015), pp. A2451--A2487.

\bibitem{sandia3}
J.P. Watson, H.J. Greenberg, and W.E. Hart, \emph{A multiple-objective analysis
  of sensor placement optimization in water networks}, in \emph{Proceedings of
  the World Water and Environment Resources Congress. American Society of Civil
  Engineers}, 2004.

\bibitem{sur_oed}
J. Yu and M. Anitescu, \emph{Multidimensional sum-up rounding for integer
  programming in optimal experimental design (preprint)}, Mathematical
  Programming  (2017).

\bibitem{yu2018scalable}
J. Yu, V.M. Zavala, and M. Anitescu, \emph{A scalable design of experiments
  framework for optimal sensor placement}, Journal of Process Control 67
  (2018), pp. 44--55.

\bibitem{zhang2006schur}
F. Zhang, \emph{The Schur complement and its applications}, Vol.~4, Springer
  Science \& Business Media, 2006.

\end{thebibliography}

\appendices
\section{Sum-up Rounding (SUR) Strategy}\label{app:sur}
The basic SUR strategy to construct a binary vector $w_{int} = (w_{int}^1,\cdots,w_{int}^n)$ from $w_{rel}=(w_{rel}^1,\cdots,w_{rel}^n)$ is given by:

\begin{equation}\label{intro_eq:roundup}
w_{int}^i=
\begin{cases}
1, & \mbox{if } \displaystyle \sum_{k=0}^i w_{rel}^i-\sum_{k=0}^{i-1} w_{int}^i\ge 0.5\\
0, & \mbox{otherwise.}
\end{cases}
\end{equation}
for $i=1,...,n$. For the extension of SUR to multiple dimensions, see a compatible two-level decomposition scheme in \citep[\S3.2]{sur_oed}.

\section{The SQP Algorithm for D-optimal Design}\label{sup:d-opt}
The algorithm is very similar to the one for A-optimal design, except that the gradient and Hessian for D-optimal design objective function are different.
\subsection*{Gradient of $\log\det$ objective}
First we find the derivatives to the $\log\det$ of $\postcov$:
\begin{equation}
\frac{\partial\log\det(\postcov)}{\partial w_i} = -tr\Big((F^TWF+I_n)^{-1}f_if_i^T \Big)= -f_i^T(F^TWF+I_n)^{-1}f_i
\end{equation}

\subsection{Hessian of $\log\det$ objective}
The $(i,j)^{th}$ entry of the Hessian matrix is
\begin{equation}
H_{ij} = \frac{\partial^2 \log\det(\postcov)}{\partial w_i\partial w_j} = \Big( f_i^T(F^TWF + I_n)^{-1} f_j \Big)^2.
\end{equation}

\subsection{Approximation of gradient and Hessian}
We give details for the one-dimensional case, and the procedure can be extended trivially to rectangle domains in multiple dimensions using tensor product. For the input domain, let $\{\bar{x_i}\}_{i=1}^N$ be the $N$ Chebyshev interpolation points, $\{x_i\}_{i=1}^n$ be the $n$ discretization points on the mesh and note $N = O(\log(n))$, and $C_x\in\mathbb{R}^{n\times N}$ be the matrix of interpolation coefficients. Similarly, we can construct $C_y$ for the output domain. We approximate $F$ by
\[
F_s = C_x^T\bar{F}C_y
\]
where $\bar{F}\in\mathbb{R}^{N\times N}$ is the matrix of $f(\bar{x}_i, \bar{x}_j)$ evaluated at interpolation points. Next, we construct $M\in\mathbb{R}^{N\times N}$ by setting its $(i,j)^{th}$ entry to be
\[
\bar{f}_i^T(F_s^TWF_s + I_n)^{-1}\bar{f}_j
\]
where $\bar{f}_i$ is the $i^{th}$ column of $C_y^TF_s^T$. The $i^{th}$ gradient is approximated by 
\[
g_i\approx c_x(x_i)^TMc_x(x_i).
\] To approximate the Hessian $H$, let $\bar{H} \in\mathbb{R}^{N\times N}$ and $\bar{H}(i, j) = M(i,j)^2$, and then
\[
H\approx C_x^T\bar{H}C_x.
\]
Once we figure out the gradient and Hessian approximations, it should be clear on the implementation of the SQP Algorithm \ref{sqp_algo:sqp} in \S\ref{sec:sqp}.

\subsection{Error analysis}
All the error analysis in \S\ref{algo_sec:error} applies to the log-determinant case, and we will only need to modify one step in Claim \ref{algo_claim:obj1}:
\begin{align*}
|\phi(w) - \phi_s(w)| &= \big|\sum_{i=1}^n\log\frac{1}{1 + \lambda^n_i} - \sum_{i=1}^n\log\frac{1}{1 + \lambda^{n,s}_i} \big| \\
& \le \big| \sum_{i=1}^N \big(\log\frac{1}{1+\lambda^n_i} - \log\frac{1}{1+\lambda^{n,s}_i} \big) \big| + \big| \sum_{i=N+1}^n \big(\log\frac{1}{1+\lambda^n_i} -\log \frac{1}{1+\lambda^{n,s}_i} \big) \big| \\
& = \sum_{i=1}^N\big|\log(1+\lambda^{n,s}_i)-\log(1+\lambda^{n}_i)\big| + \sum_{i=N+1}^N\log(1+\lambda^n_i) \\
& \le \sum_{i=1}^N|\lambda^n_i-\lambda^{n,s}_i| + \sum_{i=N+1}^n\lambda^n_i.
\end{align*}

\end{document}